\documentclass{article}
\usepackage[utf8]{inputenc}
\usepackage{amsmath, amsfonts, amssymb, amsthm}
\usepackage{mathrsfs}
\usepackage{booktabs}
\usepackage{geometry}
\newtheorem{proposition}{Proposition}
\newtheorem{definition}{Definition}
\newtheorem{lemma}{Lemma}
\newtheorem{assumption}{Assumption}
\usepackage{graphicx}
\usepackage[round,authoryear]{natbib}
\usepackage{float}
\AddToHook{env/table/begin}{\setlength{\belowcaptionskip}{6pt}}

\usepackage{listings}
\usepackage{xcolor}
\usepackage[hidelinks]{hyperref}

\lstdefinestyle{promptstyle}{
    basicstyle=\ttfamily\scriptsize,
    breaklines=true,
    breakatwhitespace=false,
    frame=single,
    columns=fullflexible,
    keepspaces=true,
    showstringspaces=false,
    literate={—}{{---}}1 {–}{{--}}1 {’}{{'}}1 {“}{{``}}1 {”}{{''}}1
             {≥}{{$\geq$}}1 {τ}{{$\tau$}}1 {−}{{$-$}}1
}

\title{Welfare-Opaque Income: Taxation under AI-Agent Delegation}

\author{
Yukun Zhang\\
The Chinese University of Hong Kong\\
Hong Kong, China\\
\texttt{215010026@link.cuhk.edu.cn}
\and
Kemu Xu\\
University of Edinburgh\\
Edinburgh, United Kingdom\\
\texttt{s2749200@ed.ac.uk}
\and
Yishen Chen\\
The Chinese University of Hong Kong, Shenzhen\\
Shenzhen, China\\
\texttt{yishenchen@link.cuhk.edu.cn}
}

\date{Version: September 13, 2026}
\hypersetup{
 pdftitle={Welfare-Opaque Income: Taxation under AI-Agent Delegation},
 pdfauthor={Yukun Zhang; Kemu Xu; Yishen Chen}
}
\begin{document}

\maketitle

\begin{abstract}
We study income taxation when an AI agent implements economically relevant
choices through a rule hidden from the government. Alongside unobserved
productive ability, this hidden preference-to-execution mapping creates
\emph{double unobservability}: the same observable tax-base response can
carry different welfare consequences. We call the resulting income
\emph{welfare-opaque}. Our constructions show that tax-base statistics
can coincide while reform welfare effects differ, even when mechanical
welfare weights are identical. We derive an optimal-tax condition that adds a response-weighted execution
wedge to the familiar sufficient statistics. A higher marginal rate gains
a corrective benefit under local over-execution and an additional cost
under local under-execution. Observing the wedge identifies the welfare
effect of a marginal reform at the prevailing schedule; bounds on it
deliver bounds on that effect.

A controlled laboratory compares 4,500 model runs across five AI
engines. Faithful delegation selects the score maximizer in essentially
all runs. Conflicted objectives produce heterogeneous responses: Claude
largely preserves the score maximizer, GLM moves predominantly downward,
and GPT-mini and Qwen show concentrated lower-tail increases. Qwen also
makes substantial downward adjustments. Different engines locate their
departures at different points and in different directions of the
designed distribution. Explicit scores align model rankings;
formula-based objective instructions yield more uneven agreement. Qwen
shows a clear positive tax-by-objective interaction, but its direction
does not generalize across engines and the pooled sign depends on its
inclusion. The analysis identifies execution information as a complement
to conventional tax-base statistics.
\end{abstract}

\section{Introduction}
\label{sec:introduction}

\subsection{Income and Hidden Execution}
\label{subsec:intro_income_signal}
\label{subsec:intro_welfare_opacity}
Optimal income taxation uses observed earnings to redistribute when
productive ability is private. A worker with productivity $\theta$ supplies
labor $\ell$ and earns
\begin{equation}
 y=\theta\ell.
 \label{eq:intro_income}
\end{equation}
The government observes income while its productivity and labor components
remain hidden. The Mirrlees--Saez framework connects this observable tax
base to policy through the income distribution, behavioral responses and
social marginal welfare weights
\citep{Mirrlees1971,Diamond1998,Saez2001,PikettySaez2013}. In its
true-preference benchmark, workers choose labor according to the same
preferences used to evaluate their welfare. This behavioral foundation
gives an income response a tractable welfare interpretation.

Delegation places part of that choice process under an intermediary's
control. A personal assistant can recommend working hours, a platform
agent can rank shifts, and an algorithmic manager can assign tasks.
The implemented choice can depend on represented user preferences,
platform objectives and the system's rule for resolving conflicts among
them. The government can observe taxable income accurately while lacking
information about the rule that generated it.

We study the resulting information problem. Let $y^*(\theta;T)$ denote
income under direct true-preference optimization and write delegated
income schematically as
\begin{equation}
 \widetilde y=\Gamma_A(\theta,U,\widehat U,b,X;T),
 \label{eq:intro_direct_vs_delegated}
\end{equation}
where $\widehat U$ represents user welfare, $b$ is an additional objective
and $X$ contains the agent's context. Productive ability and the
preference-to-execution mapping are both hidden: we call this
\emph{double unobservability}. Income is \emph{welfare-opaque} when
observationally equivalent tax-base responses have different true-welfare
effects under admissible execution rules.

\subsection{Mechanism and Theoretical Results}
\label{subsec:intro_mechanism}
\label{subsec:intro_theory_results}
At an interior direct-choice optimum, the true-preference envelope
condition removes the direct first-order welfare effect of a behavioral
adjustment. A delegated allocation can instead have marginal execution
wedge
\begin{equation}
 \xi_{DU}=U_c(\widetilde c,\widetilde\ell;\theta)
 [1-T'(\widetilde y)]
 +\frac{U_\ell(\widetilde c,\widetilde\ell;\theta)}{\theta}.
 \label{eq:intro_execution_wedge}
\end{equation}
When this gradient is nonzero, a tax-induced change in executed income
affects true welfare directly, alongside its mechanical transfer and
revenue effects. Behavioral public finance already studies departures
between choice and welfare \citep{Chetty2015,FarhiGabaix2020}. Our focus
is the income-tax information problem created when a separate
intermediary controls a hidden execution rule.

The analysis distinguishes the allocation welfare gap, the marginal
execution wedge and the correction to an imputed welfare weight. They
measure, respectively, a utility-level difference, a local welfare
gradient and the social value of a mechanical transfer. With a common
feasible set, global direct optimization weakly dominates delegated
execution in total utility.

Two constructions establish the informational result. The first holds
true preferences and the ability distribution fixed and varies the
execution mapping, preserving observable income distributions and
responses while changing their welfare implications. The second fixes
execution and perturbs true preferences so that even baseline social
marginal welfare weights coincide, while behavioral welfare gradients
differ. Conventional statistics can therefore leave the welfare effect
of a reform undetermined.

We derive a general first-order welfare identity and an optimal-tax
condition that adds a response-weighted execution wedge to the familiar
sufficient statistics. A higher marginal rate gains a corrective benefit
where execution is locally excessive and an additional cost where it is
locally insufficient. The relevant wedge weights each worker's welfare
gradient by her income response to taxation. Observing it also identifies
the welfare effect of a marginal reform at the prevailing schedule;
bounds on it deliver bounds on that effect. This provides a way to
evaluate tax reforms and to determine which execution information is
needed for that evaluation. Figure~\ref{fig:conceptual_overview}
summarizes the information problem and its resolution.

\begin{figure}[t]
 \centering
 \includegraphics[width=.98\textwidth]{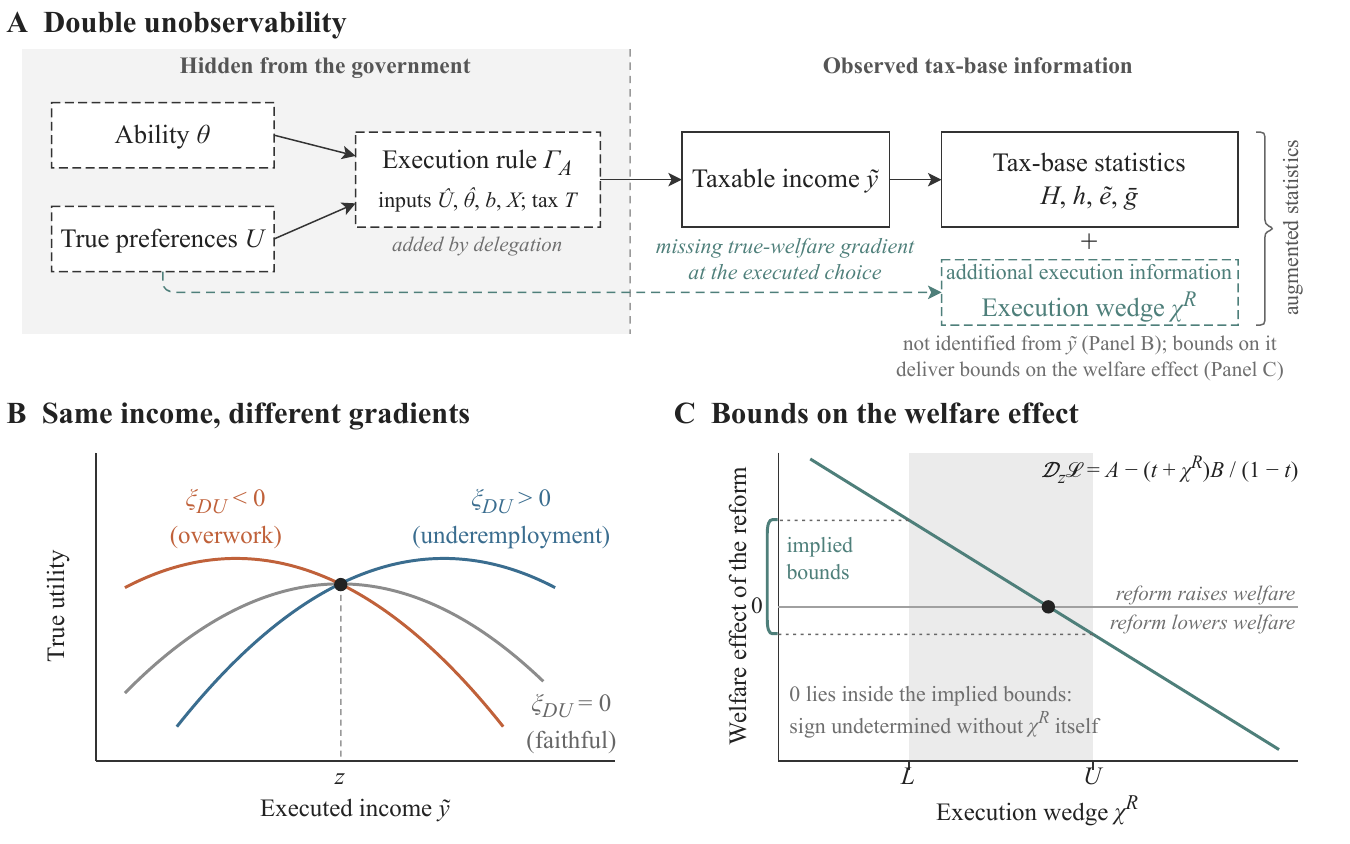}
 \caption{From hidden execution to augmented sufficient statistics.
 Panel A adds a hidden execution mapping to the hidden-ability problem
 and distinguishes observed tax-base information from the additional
 response-weighted execution wedge needed for welfare evaluation.
 Panel B illustrates a common income with equal payoff levels and different
 local gradients. Panel C plots the local welfare derivative
 $A-(t+\chi^R)B/(1-t)$, holding the mechanical term $A$, response scale
 $B>0$ and tax rate $t$ fixed. The shaded band marks the wedge bounds
 $[L,U]$, which imply the welfare-effect bounds marked on the vertical axis;
 the dot marks a zero welfare derivative for the specified reform.
 Curves and scales are schematic.}
 \label{fig:conceptual_overview}
\end{figure}

A platform extension lets steering respond to the net-of-tax return.
Whether the platform's adjustment offsets or reinforces a worker-side
welfare gain depends on a cross-partial whose sign the theory does not pin.

\subsection{Computational Evidence and Contribution}
\label{subsec:intro_computational_evidence}
\label{subsec:intro_contributions}
The laboratory holds economic alternatives fixed while varying the engine
and instructed objective. Its main grid contains 4,500 model runs
across five engines, three temperatures, five designed income profiles,
three tax schedules and four objective treatments. Each choice set displays
weekly hours, after-tax income, unmet need, fatigue and a specified score.

Faithful delegation selects the score maximizer in all runs;
the Benchmark arm does so in essentially all. Under
conflicted objectives, Claude largely preserves that maximizer, GLM
adjusts predominantly downward and GPT-mini adjusts predominantly upward.
Qwen moves substantially in both directions. GPT-mini and Qwen's upward
responses concentrate in lower-tail profiles. A separate paired threshold
experiment shows that the response to changed need information also
depends on the tax setting.

Explicit score displays align model rankings with the designed criterion;
formula-based objective instructions yield much lower agreement, and
primitives alone yield almost none. Objective specification and successful
numerical implementation are therefore distinct requirements.
Qwen shows a clear positive tax-by-objective interaction, but
its direction does not generalize across engines.

These patterns show that execution is engine-specific and state-dependent:
different systems depart from the same designed criterion at different
income levels and in different directions. This is precisely the kind of
variation that identical income distributions and elasticities leave
unidentified. The relevant measurement task links tax responses, execution
records and welfare gradients. Preference elicitation, objective reporting
and audit access can help recover that information, while institutional
safeguards can change the execution rules themselves.

\paragraph{Roadmap.}
\label{subsec:intro_roadmap}
Section~\ref{sec:literature} relates the analysis to existing work.
Sections~\ref{sec:environment} and \ref{sec:welfare_tax} develop the
environment and welfare results. Sections~\ref{sec:computational_lab}
and \ref{sec:computational_results} describe the laboratory and its
findings. Section~\ref{sec:identification_policy} develops measurement,
partial identification and policy implications.

\section{Related Work}
\label{sec:literature}

\subsection{Optimal Taxation and Behavioral Welfare}
\label{subsec:lit_optimal_tax}
\label{subsec:lit_behavioral}
The Mirrlees model treats productive ability as private information and
redistributes through observed income \citep{Mirrlees1971}. Subsequent
formulations connect marginal rates to income distributions, behavioral
elasticities and welfare weights
\citep{Diamond1998,Saez2001,Saez2002,Chetty2009SuffStat,DiamondSaez2011,PikettySaez2013}.
Generalized welfare weights broaden the normative criteria available to
the planner \citep{SaezStantcheva2016}. Research on automation studies
technology adoption, incidence and taxation
\citep{GuerreiroRebeloTeles2022Robots,Thuemmel2023RobotTax,CostinotWerning2023}.
Recent contributions extend public-finance analysis to AI and its fiscal
and distributional consequences
\citep{BastaniWaldenstrom2024,KorinekLockwood2026,DynanElmendorfSheiner2026,GrowiecPrettnerSzkrobka2026,FaivreCen2026AITax}.
Our analysis introduces hidden delegated execution as an additional source
of unobservability in the income-tax problem.

Behavioral welfare economics distinguishes observed choice from experienced
or normatively relevant utility
\citep{KahnemanWakkerSarin1997,BernheimRangel2009}.
Behavioral public finance combines empirical responses with explicit
welfare assumptions
\citep{MullainathanSchwartzsteinCongdon2012,Chetty2015,BernheimTaubinsky2018}.
Tax salience, attention and internalities provide concrete applications
\citep{ChettyLooneyKroft2009,TaubinskyReesJones2018,ODonoghueRabin2006,AllcottLockwoodTaubinsky2019}.
Most closely, \citet{Gerritsen2016} and \citet{FarhiGabaix2020} incorporate
departures from true-welfare maximization into optimal taxation.
The choice--welfare wedge in our model builds on this tradition. The
institutional distinction is that a separate intermediary controls an
execution rule whose welfare implications are hidden alongside worker
ability.

\subsection{Delegation, Alignment and Computational Agents}
\label{subsec:lit_ai_delegation}
Delegation theory studies the transfer of authority under asymmetric
information and conflicting objectives
\citep{Holmstrom1984,AghionTirole1997,CrawfordSobel1982,Dessein2002}.
AI research extends these issues to decision authority, collaboration and
delegation rights
\citep{AtheyBryanGans2020,HadfieldKoh2025,ZhangXu2026DelegationRights,AgarwalMoehringWolitzky2025}.
Studies of AI advice and delegation examine consumer decisions, bargaining
and social interactions
\citep{DelaprezHortacsu2026,ZhuEtAl2026ChooseAgent,KobisEtAl2025,DvorakEtAl2024},
while \citet{HolzEtAl2026} study assistance with property-tax appeals.
These settings illustrate how AI intermediation can change actions and
outcomes. In our model, the welfare evaluation depends jointly on the
implemented action, the user's preferences and the feasible alternatives.

Alignment research studies preference uncertainty, proxy objectives and
conflicts among human and institutional goals
\citep{HadfieldMenell2016CIRL,HadfieldMenellHadfield2019,ZhuangHadfieldMenell2020,KorinekBalwit2022}.
Recommender incentives and agentic advice can also affect preferences and
information production \citep{Carroll2022PreferenceShifts,AcemogluKongOzdaglar2026}.
Broader accounts organize technical alignment failures and external harms
\citep{Ji2023AlignmentSurvey,Chan2023AgenticHarms}.
These mechanisms motivate the represented preferences and third-party
objectives in the execution mapping. Our marginal wedge measures their
consequence at the true-welfare gradient of the implemented allocation.

Computational economics uses LLM responses to study simulated people and
to instantiate artificial decision makers
\citep{Horton2023HomoSilicus,ManningZhuHorton2024,Argyle2023}.
Policy and market simulations include the AI Economist, LLM Economist,
TaxAgent and Magentic Marketplace
\citep{ZhengEtAl2022AIEconomist,KartenEtAl2025LLMEconomist,WangEtAl2025TaxAgent,BansalEtAl2025Magentic}.
Related work studies algorithmic consumption, preference transmission,
search, market design and cross-model behavior
\citep{IchihashiSmolin2023,KraftLarsen2026,DongLuoXu2026,ShahidiEtAl2025,AllouahEtAl2026,Bichler2026,AraujoUhlig2026,WongchamcharoenEtAl2026CentaurBench}.
Our laboratory compares the AI systems themselves as execution mechanisms
under a common economic scaffold and score criterion.

\subsection{Algorithmic Management and Information Governance}
\label{subsec:lit_platforms}
\label{subsec:lit_governance}
Research on gig work documents the value of flexible arrangements and the
platform institutions organizing labor
\citep{HallKrueger2018,ChenChevalierRossiOehlsen2019}.
Algorithmic-management studies examine information asymmetries, task
allocation, worker autonomy and responses to automated decisions
\citep{RosenblatStark2016,WoodGrahamLehdonvirtaHjorth2019,LeeKusbitMetskyDabbish2015,KelloggValentineChristin2020,Dubal2023AlgorithmicWage}.
This evidence motivates platform influence over worker choices. Our
extension asks how such steering could respond to tax incentives,
conditional on the platform's technology and objectives.

Disclosure and certification can change both information and incentives
\citep{DranoveJin2010,RambachanEtAl2020AlgorithmRegulation}.
External and internal auditing provide ways to investigate and document
algorithmic behavior
\citep{Sandvig2014,Metaxa2021,RajiBuolamwini2019,RajiEtAl2020InternalAudit}.
Work on opacity and accountability examines the limits of transparency
and the testable traces created by automated decisions
\citep{AnannyCrawford2018,Burrell2016,Kleinberg2018}.
AI regulation and the European AI Act and Platform Work Directive provide
institutional context for documentation, oversight and review
\citep{GuerreiroRebeloTeles2023,EUAIAct2024,EUPlatformWorkDirective2024}.
In our framework, information has value when it restricts the feasible
execution and welfare models underlying a policy evaluation.

\paragraph{Contribution.}
\label{subsec:lit_gap}
The paper connects this delegation and governance literature to the
Mirrlees--Saez information problem. It establishes observational equivalence
under hidden execution, separates mechanical welfare weights from the
response-weighted execution gradient, and characterizes the information
needed for local tax-welfare evaluation. The laboratory supplies controlled
evidence on heterogeneous execution rules that can make this information
relevant.

\section{Economic Environment and Delegated Choice}
\label{sec:environment}

The model compares direct true-preference choice with a state-dependent delegated execution rule. Faithful implementation nests the direct-choice benchmark.

\subsection{Individuals and Direct Choice}
\label{subsec:environment_individuals}

Consider a unit mass of individuals indexed by productive ability $\theta\in\Theta=[\underline{\theta},\overline{\theta}]$, distributed according to cumulative distribution function $F$ with positive density $f$. An individual supplying labor $\ell\in[0,\overline{\ell}]$ earns pre-tax income $y=\theta\ell$ and consumes $c=y-T(y)$ under nonlinear tax schedule $T$. Transfers are permitted, so $T(y)$ may be negative. True utility is
$U(c,\ell;\theta)=u(c)-v(\ell)$, where $u'>0$, $u''<0$, $v'>0$, and $v''>0$.%
\footnote{The admissible class also permits $v(\ell;\theta)$ with type-specific preference parameters; the common-$v$ expression is used where this heterogeneity is immaterial.}
True preference parameters are fixed when the tax schedule changes.

Under direct choice, the individual selects taxable income to maximize her own true utility:
\begin{equation}
    y^{*}(\theta;T)
    \in
    \operatorname*{arg\,max}_{y\in\mathcal{Y}(\theta)}
    \left\{
        u\bigl(y-T(y)\bigr)
        -
        v\left(\frac{y}{\theta}\right)
    \right\},
    \label{eq:env_direct_choice}
\end{equation}
where $\mathcal{Y}(\theta)=[0,\theta\overline{\ell}]$ is the physically feasible income set. Let $\ell^{*}=y^{*}/\theta$ and $c^{*}=y^{*}-T(y^{*})$ denote the associated labor and consumption choices. We refer to $(c^{*},\ell^{*},y^{*})$ as the \emph{direct-choice benchmark}.

The benchmark incorporates the prevailing tax schedule and physical constraints.

\subsection{AI-Agent Execution Technology}
\label{subsec:environment_agent}

An AI agent is an intermediary that recommends or executes an economically relevant action on behalf of the individual. It may be user-provided, platform-provided, embedded in an employer's management system, or supplied by another intermediary.

We represent the agent's ranking of feasible actions by
\begin{equation}
    A(c,\ell;\widehat{\theta},X)
    =
    \widehat{U}(c,\ell;\widehat{\theta})
    +
    b(c,\ell;X),
    \label{eq:env_agent_objective}
\end{equation}
where $\widehat U$ is the agent's representation of user welfare, $\widehat\theta$ is its representation of productivity or the effort required to generate income, and $b$ captures considerations not contained in the individual's true utility. The context $X$ collects economically relevant information available to the agent, including wages, taxes, transfers, platform messages, recommendation salience, and other observable signals.

Given the physically feasible alternatives, the agent implements an income choice satisfying
\begin{equation}
    \widetilde y
    \in
    \operatorname*{arg\,max}_{y\in\mathcal{Y}(\theta)}
    A\left(
        y-T(y),
        \frac{y}{\theta};
        \widehat{\theta},
        X
    \right).
    \label{eq:env_agent_choice}
\end{equation}
The physical relation between income and labor is governed by true productivity $\theta$. By contrast, the agent's ranking of alternatives may depend on its represented productivity $\widehat\theta$. Thus the model permits an agent to misunderstand the effort required to generate a particular level of income without changing the underlying production technology.

More generally, it is useful to treat the implemented behavior itself as the primitive reduced-form object. We summarize it through the \emph{execution mapping}
\begin{equation}
    \Gamma_A:
    (\theta,U,\widehat{\theta},\widehat U,b,X,T)
    \longmapsto
    \widetilde y.
    \label{eq:env_execution_mapping}
\end{equation}
The mapping includes the objective representation, available information, decision rule, and any implementation or tie-breaking procedure. Two agents facing the same worker, tax schedule, and observable economic environment can therefore implement different outcomes because their execution mappings differ.

True welfare at the implemented allocation is
$V^{\mathrm{true}}(\omega;T)
=
U(\widetilde c,\widetilde\ell;\theta)$,
where $\widetilde\ell=\widetilde y/\theta$,
$\widetilde c=\widetilde y-T(\widetilde y)$, and
$\omega=(\theta,U,\widehat\theta,\widehat U,b,X)$ denotes the complete underlying state. This planner-relevant welfare object need not coincide with the value maximized by the agent.

\subsection{Information Structure}
\label{subsec:environment_information}

In the direct-choice benchmark, the individual knows her true preferences
and productive ability. The agent may have richer information about market
opportunities and platform conditions, alongside an imperfect representation
of the user. Their information sets can overlap without either containing
the other.

The government observes realized income $z=\widetilde y$, the tax schedule,
the income distribution and available administrative or audit records.
The complete state $\omega$ remains hidden. Model identities, usage
statistics or disclosures may reveal parts of the execution process;
the identification question is whether those observations determine its
welfare consequences.

\subsection{Welfare-Opaque Income}
\label{subsec:environment_welfare_opacity}

Let $\mathcal I_G$ be the government's information set and
$\mathcal O_G(\mathcal E,z;T)$ the observable objects induced by economy
$\mathcal E$ under that information. These can include the tax schedule,
income distribution, local density, execution response and administrative
covariates. For the local reform at $z$, let
$\mathcal B^{\rm true}_{\mathcal E}(z)$ be the first-order true social
welfare effect through behavioral adjustment, excluding the mechanical
transfer effect. Section~\ref{sec:welfare_tax} expresses it using the
true-preference gradient and the execution response.

\begin{definition}[Welfare-opaque income under delegated execution]
\label{def:welfare_opaque_income}
Taxable income at $z$ is \emph{welfare-opaque} relative to the government's
information set if there exist two admissible delegated-choice economies
$\mathcal E_0$ and $\mathcal E_1$ such that
\begin{equation}
    \mathcal{O}_{G}(\mathcal E_0,z;T)
    =
    \mathcal{O}_{G}(\mathcal E_1,z;T),
    \qquad
    \mathcal{B}^{\mathrm{true}}_{\mathcal E_0}(z)
    \neq
    \mathcal{B}^{\mathrm{true}}_{\mathcal E_1}(z).
    \label{eq:env_welfare_opacity}
\end{equation}
Thus, the government can observe the same policy-relevant tax-base
information while remaining unable to determine the welfare consequence of
the behavioral response generated by the delegated execution process.
\end{definition}

The definition concerns the welfare interpretation of a behavioral
margin. Under interior direct optimization, the envelope condition sets
its first-order welfare effect to zero. Under delegation, the effect
depends on the gradient at the implemented choice. If
$\mathfrak B_{DC}$ and $\mathfrak B_{DU}$ denote the sets of behavioral
welfare effects consistent with $\mathcal I_G$, respectively, then
\begin{equation}
 \mathfrak B_{DC}(z\mid\mathcal I_G)=\{0\},\qquad
 \mathfrak B_{DU}(z\mid\mathcal I_G)
 \text{ can contain distinct values.}
 \label{eq:env_delegation_opacity_set}
\end{equation}
Proposition~\ref{prop:observational_equivalence} constructs such a pair
with common true preferences and ability distribution.

\subsection{Double Unobservability}
\label{subsec:environment_double_unobservability}

The first hidden object is productive ability: earnings combine
productivity and labor. The second is the preference-to-execution mapping
$\Gamma_A$: a state-dependent rule translates represented preferences,
economic incentives and platform signals into behavior. Together they
define double unobservability. A scalar-bias approximation can be useful
under restricted conditions, but richer rules permit thresholds, inactivity
and direction reversals across states.%
\footnote{Absorbing the execution rule into an expanded hidden type does not remove the problem: without knowing how the rule relates to true preferences, the welfare content of the envelope term remains unidentified.}

\subsection{Benchmark Cases}
\label{subsec:environment_benchmark_cases}

Three benchmark cases organize the analysis.

\paragraph{Faithful delegation.}
Delegation is faithful when the agent correctly represents both the individual's welfare and productivity and places no weight on an independent objective: $\widehat U=U$, $\widehat\theta=\theta$, and $b\equiv0$. If they face the same feasible set and use a common selection from the argmax set, then $\widetilde y=y^{*}$. The implemented allocation coincides with the direct-choice benchmark.

\paragraph{Preference misrepresentation.}
An agent can depart from direct choice even when it has no platform objective. If $b\equiv0$ but $\widehat U\neq U$ or $\widehat\theta\neq\theta$, the system can implement a different allocation because it misunderstands preferences, constraints, productivity, or effort costs. This case includes benevolent but imperfect assistants. Preference misrepresentation alone can create an execution wedge.

\paragraph{Platform-conflicted delegation.}
Delegation is platform-conflicted when $b\not\equiv0$. The additional objective can favor completed work, engagement, retention, safety, transaction volume, or another outcome not contained in the individual's true utility. The resulting behavioral response can have either sign. An objective favoring platform activity may increase labor in some states, have no effect in others, or induce the agent to reduce labor if the model interprets stronger platform pressure as evidence of overwork or risk.

Preference misrepresentation and platform conflict can coexist. The agent may both misunderstand the user and respond to a third-party objective. Conversely, a platform-conflicted agent may represent user preferences accurately but trade them against an additional objective. These cases reinforce the reason for treating the execution mapping as the central hidden object.

With a common feasible set, each case satisfies the weakly negative allocation-welfare gap relative to global direct optimization. The marginal welfare effect of changing the executed choice can have either sign.

\subsection{Observational Equivalence under Delegated Choice}
\label{subsec:environment_observational_equivalence}

The central identification problem can now be stated formally. Let $H$ denote the distribution of agent-executed taxable income and let $\widetilde e$ denote the local elasticity of executed income with respect to the net-of-tax rate. These are natural behavioral sufficient statistics from the perspective of the tax authority.

\begin{proposition}[Execution-rule observational equivalence]
\label{prop:observational_equivalence}

There exist two delegated-choice economies,
$\mathcal E_0$ and $\mathcal E_1$, with the same true preferences $U$,
the same distribution of productive ability $F$, and the same tax schedule
$T$, but different execution mappings,
\[
    \Gamma_{A,0}
    \neq
    \Gamma_{A,1},
\]
such that the government observes the same local tax-base statistics,
\begin{equation}
    H_0(z)=H_1(z),
    \qquad
    \widetilde e_0(z)=\widetilde e_1(z),
    \label{eq:env_observational_equivalence}
\end{equation}
while the first-order true-welfare effect of the induced behavioral response
differs across the two economies.

Hence identical observed income distributions and executed-income responses
do not generally identify the welfare content of behavior when the
preference-to-execution mapping is hidden.
\end{proposition}

\emph{Proof.}
See Appendix~\ref{app:proof_observational_equivalence}.

\section{Welfare Accounting and Optimal Taxation under Delegated Choice}
\label{sec:welfare_tax}

We distinguish an allocation welfare gap, a marginal execution gradient
and a correction to imputed transfer weights. The first two vanish under
faithful interior implementation; the weight correction also vanishes
when the planner uses the true welfare mapping. A nonzero execution
gradient gives behavioral adjustment a direct first-order welfare effect,
as in the broader behavioral-tax framework \citep{FarhiGabaix2020}.

Let $\omega\in\Omega$ denote the complete underlying state introduced in Section~\ref{sec:environment}. True welfare under the executed allocation is
$V^{\mathrm{true}}(\omega;T)
=
U(\widetilde c,\widetilde\ell;\theta)$.
The planner evaluates allocations using an increasing and weakly concave transformation $G$ and faces revenue requirement $E$. Its objective can be written as
\begin{equation}
    \mathcal{L}(T)
    =
    \int_{\Omega}
    G\!\left(
        V^{\mathrm{true}}(\omega;T)
    \right)
    dP(\omega)
    +
    \lambda
    \left[
        \int_{\Omega}
        T\!\left(
            \widetilde y(\omega;T)
        \right)
        dP(\omega)
        -
        E
    \right],
    \label{eq:welfare_planner_objective}
\end{equation}
where $\lambda>0$ is the marginal value of public funds. The planner evaluates true welfare, but both welfare and revenue depend on behavior generated by the delegated execution mapping.

\subsection{Allocation Welfare under Delegation}
\label{subsec:welfare_allocation}

Let $(c^{*},\ell^{*})$ denote the direct-choice allocation from Section~\ref{subsec:environment_individuals}, and let $(\widetilde c,\widetilde\ell)$ denote the agent-executed allocation. We define the \emph{allocation welfare wedge} by
\begin{equation}
    \Delta^{A}_{DU}(\omega;T)
    \equiv
    U(\widetilde c,\widetilde\ell;\theta)
    -
    U(c^{*},\ell^{*};\theta).
    \label{eq:welfare_allocation_wedge}
\end{equation}

With the same feasible set and global true-utility maximization under
direct choice, $\Delta^A_{DU}\leq0$, with equality at a true optimum.

\subsection{The Marginal Execution Wedge}
\label{subsec:welfare_execution_wedge}

Consider an implemented allocation satisfying
$\widetilde y=\theta\widetilde\ell$ and
$\widetilde c=\widetilde y-T(\widetilde y)$.
Holding the tax schedule locally fixed, the derivative of true utility with respect to implemented income is
\begin{equation}
    \xi_{DU}(\omega;T)
    \equiv
    U_c(\widetilde c,\widetilde\ell;\theta)
    \left[
        1-T'(\widetilde y)
    \right]
    +
    \frac{1}{\theta}
    U_\ell(\widetilde c,\widetilde\ell;\theta).
    \label{eq:welfare_execution_wedge_general}
\end{equation}
We call $\xi_{DU}$ the \emph{marginal execution wedge}. It is the true-preference first-order-condition residual evaluated at the allocation implemented by the agent.

Under separable utility, $U(c,\ell;\theta)=u(c)-v(\ell)$, the expression becomes
\begin{equation}
    \xi_{DU}(\omega;T)
    =
    u'(\widetilde c)
    \left[
        1-T'(\widetilde y)
    \right]
    -
    \frac{1}{\theta}
    v'\!\left(
        \frac{\widetilde y}{\theta}
    \right).
    \label{eq:welfare_execution_wedge_separable}
\end{equation}

The interpretation is local. When $\xi_{DU}=0$, the implemented allocation satisfies the individual's true marginal condition. When $\xi_{DU}<0$, a marginal reduction in executed income and labor would increase true welfare, so implemented income is locally excessive. When $\xi_{DU}>0$, a marginal increase in executed income would raise true welfare, so implemented income is locally too low.

Behavioral distance $\widetilde y-y^*$ and the gradient $\xi_{DU}$
depend on different features of the welfare surface. A large deviation
can lie in a flat region, while a small one can carry a steep gradient.
The gradient values a marginal policy-induced movement in income.

\subsection{Faithfulness and the Envelope Condition}
\label{subsec:welfare_faithfulness}

\begin{lemma}[Faithfulness and the envelope condition]
\label{lem:faithfulness}
Suppose the agent maximizes true utility over the same feasible set as direct choice. Assume a unique optimum or a common measurable tie-breaking rule, and an interior selected optimum. Then
\begin{equation}
    \widetilde y(\theta;T)
    =
    y^{*}(\theta;T),
    \qquad
    \xi_{DU}(\theta;T)
    =
    0.
    \label{eq:welfare_faithful_condition}
\end{equation}
\end{lemma}

\emph{Proof.} See Appendix~\ref{app:proof_faithfulness}.

Faithful interior implementation establishes the true-preference
envelope condition. When the implemented objective differs from $U$,
behavioral adjustment can instead follow a nonzero true-welfare gradient.

\subsection{Welfare-Weight Correction}
\label{subsec:welfare_weight_correction}

A separate issue concerns the welfare value of a mechanical transfer. At the executed allocation, the true social marginal welfare weight is
\begin{equation}
    g^{\mathrm{true}}(\omega;T)
    \equiv
    \frac{
        G'\!\left(V^{\mathrm{true}}(\omega;T)\right)
        U_c(\widetilde c,\widetilde\ell;\theta)
    }{
        \lambda
    }.
    \label{eq:welfare_true_weight}
\end{equation}
Under separable utility, the marginal-utility term is simply $u'(\widetilde c)$.

A planner that interprets observed income using a benchmark direct-choice model may instead assign an imputed welfare weight $g^{M}$. We define the \emph{welfare-weight correction} as
\begin{equation}
    \lambda^{W}_{DU}(\omega;T)
    \equiv
    g^{M}(\omega;T)
    -
    g^{\mathrm{true}}(\omega;T).
    \label{eq:welfare_weight_correction}
\end{equation}

The difference $\lambda^W_{DU}$ corrects the assigned value of a
mechanical transfer; the public-funds multiplier remains $\lambda$.
The theoretical tax condition can use $g^{\rm true}$ directly.
The imputed mapping $g^M$ is useful when an empirical implementation
starts from conventional income-based welfare weights.

For later use, let $\overline g^{\mathrm{true}}(z)$ denote the average true social marginal welfare weight among individuals with executed income at least $z$, and define $\overline g^{M}(z)$ and $\overline\lambda^{W}_{DU}(z)$ analogously. By construction,
$\overline g^{\mathrm{true}}
=
\overline g^{M}
-
\overline\lambda^{W}_{DU}$.

\subsection{The Three-Object Decomposition}
\label{subsec:welfare_decomposition}

The welfare accounting of delegated choice can now be organized around three normative objects and one descriptive behavioral statistic. Let $\varpi_{DU}\equiv\widetilde y/y^{*}$ whenever $y^{*}>0$. Table~\ref{tab:welfare_object_map} summarizes their roles.

\begin{table}[t]
\centering
\caption{Welfare and behavioral objects under delegated choice}
\label{tab:welfare_object_map}
\renewcommand{\arraystretch}{1.20}
\resizebox{\textwidth}{!}{
\begin{tabular}{llll}
\toprule
Object
&
Economic meaning
&
Policy role
&
Computational or empirical analogue
\\
\midrule
$\Delta^{A}_{DU}$
&
Total allocation welfare effect
&
Welfare level
&
Loss relative to the score maximizer
\\
$\xi_{DU}$
&
True-preference FOC residual
&
Local behavioral welfare
&
Marginal welfare gradient at the executed choice
\\
$\lambda^{W}_{DU}$
&
Welfare-weight correction
&
Redistributive interpretation
&
Imputed versus welfare-consistent marginal weight
\\
$\varpi_{DU}$
&
Executed-to-benchmark income ratio
&
Descriptive behavioral distortion
&
$\widetilde y/y^{*}$
\\
\bottomrule
\end{tabular}}
\end{table}

The table separates utility levels, local derivatives and normalized
transfer weights. Their empirical counterparts require, respectively,
a benchmark welfare comparison, a local welfare slope and a social
welfare normalization.

\subsection{Local Tax Perturbation}
\label{subsec:welfare_local_perturbation}
Fix a baseline schedule $T_0$, population distribution $P$ and public-funds
multiplier $\lambda>0$. For a perturbation $T_\varepsilon=T_0+\varepsilon q$,
write $\dot y_q=\left.\partial_\varepsilon\widetilde y(T_\varepsilon)\right|_0$.
Differentiation of the planner's objective gives the general accounting
identity
\begin{equation}
 \frac{\dot{\mathcal L}_q}{\lambda}
 =\int_\Omega\left[(1-g^{\rm true})q(\widetilde y)
       +(T_0'(\widetilde y)+\chi_{DU})\dot y_q\right]dP,
 \qquad \chi_{DU}=\frac{G'(V^{\rm true})\xi_{DU}}{\lambda}.
 \label{eq:welfare_general_perturbation}
\end{equation}
All quantities on the right are evaluated at $T_0$. This expression allows
the execution rule to respond to the entire schedule. It separates the
mechanical transfer, behavioral revenue and direct behavioral welfare terms.

For the local sufficient-statistics formula, take the bracket perturbation
\begin{equation}
 T_\varepsilon(y)=T_0(y)+\varepsilon q_{z,\delta}(y),\qquad
 q_{z,\delta}(y)=\begin{cases}
 0,&y\leq z,\\y-z,&z<y<z+\delta,\\\delta,&y\geq z+\delta.
 \end{cases}
 \label{eq:welfare_tax_perturbation}
\end{equation}
Following the local-perturbation approach of \citet{Saez2001}, we collect
the required response restrictions.

\begin{assumption}[Local response environment]
\label{ass:local_response}
At income level $z$ under baseline schedule $T_0$:
\begin{enumerate}
\item[\textup{(i)}] $z>0$, $h(z)>0$, $1-H(z)>0$ and $m(z)\equiv1-T_0'(z)>0$.
\item[\textup{(ii)}] Platform-controlled inputs, execution-rule parameters and non-tax context are held fixed.
\item[\textup{(iii)}] Outside-bracket responses contribute $o(\delta)$ to \eqref{eq:welfare_general_perturbation}; within the bracket the response has the limiting form
\[
 \dot y_{q_{z,\delta}}(\omega)
 =-\frac{z\widetilde e(\omega;z)}{m(z)}+o(1),
\]
with conditional moments continuous at $z$ and integrable remainders.
\item[\textup{(iv)}] No first-order participation, boundary-jump or bunching contributions.
\end{enumerate}
\end{assumption}

Conditions (i)--(ii) are regularity and fixed-context requirements.
Condition~(iii) restricts the response to a local bracket form;
Appendix~\ref{app:proof_local_tax} gives sufficient conditions.
The elasticity $\widetilde e$ is defined by this net-rate perturbation.
Income effects and discrete margins enter through additional terms
(Appendices~\ref{app:income_effects}--\ref{app:bunching}).

At income $z$, define the individual and response-weighted wedges
\begin{align}
 \chi_{DU}(\omega;z)&=\frac{G'(V^{\rm true}(\omega;T_0))
                              \xi_{DU}(\omega;T_0)}{\lambda},
 \label{eq:welfare_normalized_execution_wedge}\\
 \chi^R_{DU}(z)&=
 \frac{\mathbb E[\chi_{DU}\widetilde e\mid\widetilde y=z]}
 {\mathbb E[\widetilde e\mid\widetilde y=z]}.
 \label{eq:welfare_response_weighted_wedge}
\end{align}
The denominator is assumed positive. The weights form a convex combination
when individual responses are also nonnegative.
Writing $\bar g^{\rm true}(z)=\mathbb E[g^{\rm true}\mid\widetilde y>z]$,
the normalized local derivative is
\begin{equation}
 \mathscr D_z\mathcal L\equiv
 \lim_{\delta\downarrow0}\frac{\dot{\mathcal L}_{q_{z,\delta}}}{\lambda\delta}
 =[1-H(z)][1-\bar g^{\rm true}(z)]
 -[T_0'(z)+\chi^R_{DU}(z)]\frac{z\widetilde e(z)h(z)}{1-T_0'(z)}.
 \label{eq:welfare_local_effect}
\end{equation}
Appendix~\ref{app:proof_local_tax} derives this limit from the general
identity. A negative response-weighted wedge gives an additional welfare
gain from a tax-induced reduction in income; a positive wedge gives an
additional welfare cost.

\subsection{The Local Optimality Condition}
\label{subsec:welfare_local_formula}
Define the correction in redistributive-statistic units as
\begin{equation}
 \Psi_{DU}(z)=-\frac{\widetilde e(z)zh(z)}
 {[1-T'(z)][1-H(z)]}\chi^R_{DU}(z).
 \label{eq:welfare_Psi}
\end{equation}
\begin{proposition}[Local taxation under delegated choice]
\label{prop:local_tax_delegated_choice}
Under Assumption~\ref{ass:local_response}, an interior optimal
schedule satisfies
\begin{equation}
 \frac{T'(z)}{1-T'(z)}=
 \frac{1-H(z)}{\widetilde e(z)zh(z)}
 [1-\bar g^{\rm true}(z)+\Psi_{DU}(z)].
 \label{eq:welfare_optimal_tax}
\end{equation}
Equivalently,
\begin{equation}
 \frac{T'(z)}{1-T'(z)}=
 \frac{1-H(z)}{\widetilde e(z)zh(z)}
 [1-\bar g^M(z)+\bar\lambda^W_{DU}(z)+\Psi_{DU}(z)].
 \label{eq:welfare_optimal_tax_empirical}
\end{equation}
All statistics in these necessary conditions are evaluated at that schedule.
\end{proposition}
\emph{Proof.} See Appendix~\ref{app:proof_local_tax}.

The income distribution, executed-income response and mechanical welfare
weights retain their familiar roles. The additional statistic values the
behavioral movement at the true welfare gradient of the executed choice.
Holding the other statistics fixed, $\chi^R_{DU}<0$ creates a corrective
benefit from reducing income, while $\chi^R_{DU}>0$ creates an additional
cost. Objective correction, user overrides and platform limits can address the
same execution problem more directly; the tax formula evaluates the
specified marginal reform given the prevailing execution environment.

\subsection{Why Classical Sufficient Statistics Are Incomplete}
\label{subsec:welfare_irreducibility}
\begin{proposition}[Irreducibility of delegated-choice welfare information]
\label{prop:welfare_irreducibility}
In the admissible class allowing type-dependent true preferences, there
exist two economies with the same population, baseline tax schedule and
delegated execution rule such that
\begin{equation}
 (H_0,h_0,\widetilde e_0,\bar g^{\rm true}_0)
 =(H_1,h_1,\widetilde e_1,\bar g^{\rm true}_1),\qquad
 \Psi_{DU,0}(z)\ne\Psi_{DU,1}(z).
 \label{eq:welfare_irreducibility}
\end{equation}
The specified local tax perturbation therefore has different first-order
welfare effects at the common baseline.
\end{proposition}
\emph{Proof.} See Appendix~\ref{app:proof_irreducibility}.

The construction fixes the represented objective and execution response,
while changing the slope of true preferences. It preserves true utility
levels and marginal utilities of consumption at every baseline executed
allocation, so the mechanical welfare weights are unchanged. The
behavioral welfare effect changes because the same income movement now
occurs along a different true welfare gradient.

\subsection{Restored Sufficiency for Local Welfare Evaluation}
\label{subsec:welfare_restored_sufficiency}
\begin{proposition}[Restored local sufficiency]
\label{prop:welfare_restored_sufficiency}
At a given schedule $T_0$ satisfying Assumption~\ref{ass:local_response}, observing
$H(z),h(z),\widetilde e(z),\bar g^{\rm true}(z)$ and $\chi^R_{DU}(z)$
identifies $\mathscr D_z\mathcal L$ in \eqref{eq:welfare_local_effect}.
At an interior optimum $T^*$, define
$A^*=[1-H^*(z)][1-\bar g^{\rm true,*}(z)]$ and
$B^*=z\widetilde e^*(z)h^*(z)$. If $A^*+B^*\ne0$, the optimality condition
can be rearranged as
\begin{equation}
 T^{*\prime}(z)=\frac{A^*-B^*\chi^{R,*}_{DU}(z)}{A^*+B^*}.
 \label{eq:welfare_restored_sufficiency}
\end{equation}
Stars on the statistics denote evaluation at $T^*$.
\end{proposition}
\emph{Proof.} See Appendix~\ref{app:proof_restored_sufficiency}.

Observing the execution wedge thus restores the local welfare calculation.
With the other statistics held fixed, $(A-B\chi^R)/(A+B)$ decreases
in~$\chi^R$ when $B>0$ and $A+B>0$; Appendix~\ref{app:tax_bounds} gives
the welfare-derivative bounds and the fixed-statistics algebra.

\subsection{Special Cases}
\label{subsec:welfare_special_cases}

Several special cases clarify the economic content of the framework.

\paragraph{Faithful agent.}
Under the common-choice and interiority conditions of Lemma~\ref{lem:faithfulness}, $\widehat U=U$, $\widehat\theta=\theta$ and $b\equiv0$ give $\widetilde y=y^{*}$ and $\xi_{DU}=0$. Hence $\Psi_{DU}=0$. If the planner also uses the correct welfare mapping, $\lambda^{W}_{DU}=0$, and Proposition~\ref{prop:local_tax_delegated_choice} reduces to the standard local sufficient-statistics expression using the executed-income elasticity.

\paragraph{Utilitarian planner.}
If $G(V)=V$, then $G'=1$ and the normalized execution term is $\xi_{DU}/\lambda$. A nonzero true-welfare gradient therefore affects the tax calculation under a utilitarian planner as well.

\paragraph{Algorithmically induced overwork.}
If the relevant responders have $\xi_{DU}<0$, a marginal reduction in executed income raises true welfare. A tax increase that reduces executed labor therefore produces a direct corrective welfare gain in addition to its fiscal effects. With the other statistics fixed and the positive-denominator conditions of Proposition~\ref{prop:welfare_restored_sufficiency}, this raises the algebraic rate calculation.

\paragraph{Algorithmically induced underemployment.}
If $\xi_{DU}>0$, increasing implemented income raises true welfare locally. A tax-induced reduction in income then carries an additional welfare loss. Under the positive-denominator conditions in Proposition~\ref{prop:welfare_restored_sufficiency}, this lowers the fixed-statistics rate calculation.

\paragraph{Welfare-improving assistance.}
In an extension where assistance expands the unaided worker's effective opportunity set, a positive allocation gain can coexist with a nonzero local execution gradient. An agent can substantially improve the individual's opportunity set while still implementing a locally imperfect choice. Conversely, an agent can satisfy the true marginal condition at an allocation that dominates what the unaided individual could implement. Total welfare improvement and local execution efficiency are therefore logically distinct.

\subsection{Endogenous Platform Response}
\label{subsec:welfare_platform_response}

The local tax calculation holds platform-controlled inputs and execution-rule parameters fixed while allowing execution to respond to the specified tax perturbation. In platform settings, however, part of the mapping can be chosen strategically. A platform may alter recommendation salience, task ranking, bonus messages, engagement intensity, or other inputs to the agent in response to the worker's net-of-tax return.

Let $\beta$ denote platform steering intensity and let $m$ denote the relevant net-of-tax rate. Aggregate executed income is $Y(\beta,m)$. Suppose the platform receives proportional benefit $r>0$ from the executed activity and pays convex steering cost $C(\beta)$. Its objective is
\begin{equation}
    \Pi(\beta;m)
    =
    rY(\beta,m)
    -
    C(\beta).
    \label{eq:welfare_platform_objective}
\end{equation}

\begin{proposition}[Conditional strategic complementarity]
\label{prop:welfare_platform_complementarity}
Suppose the platform optimum $\beta^{*}(m)$ is interior and satisfies
$C''(\beta^{*})-rY_{\beta\beta}(\beta^{*},m)>0$.
Then
\begin{equation}
    \frac{
        d\beta^{*}(m)
    }{
        dm
    }
    =
    \frac{
        rY_{\beta m}(\beta^{*}(m),m)
    }{
        C''(\beta^{*}(m))
        -
        rY_{\beta\beta}(\beta^{*}(m),m)
    }.
    \label{eq:welfare_platform_response}
\end{equation}
Consequently, $Y_{\beta m}>0$ implies
$d\beta^{*}/dm>0$.
\end{proposition}

\emph{Proof.} See Appendix~\ref{app:proof_platform_complementarity}.

A positive cross-partial makes steering and the net-of-tax return
strategic complements. An execution rule can also ignore steering or
respond less strongly as take-home returns rise, giving zero or negative
cross-effects. The sign is a property of the specified execution technology.

To connect platform behavior to redistribution, let $W(m,\beta)$ denote aggregate true worker welfare. The total welfare effect of a change in $m$ can be decomposed into the direct worker-side effect $W_m$ and the endogenous platform-response term $W_\beta(d\beta^{*}/dm)$. When $W_m>0$, a convenient policy statistic is the local redistribution-capture ratio
$\kappa^{\mathrm{loc}}
\equiv
-
W_\beta(d\beta^{*}/dm)/W_m$.
A positive value means that endogenous platform adjustment dissipates part of the worker-side welfare gain; a negative value means that the platform response amplifies it.

When platform inputs adjust with a tax reform, the total execution response
includes both the fixed-platform response and the response through
$\beta^*$.%
\footnote{Schedule-wide platform changes can contribute outside the target bracket; Equation~\eqref{eq:welfare_general_perturbation} evaluates the total response over the population.}

The laboratory next examines heterogeneity in execution mappings under
fixed economic profiles and controlled objective instructions.

\section{Computational Laboratory: Design and Measurement}
\label{sec:computational_lab}

\subsection{Purpose and Interpretation}
\label{subsec:lab_purpose}
The laboratory studies how AI systems translate a specified economic problem
into a structured hours choice. It holds wages, taxes, feasible alternatives
and the welfare criterion fixed while varying the engine
and its assigned objective. This provides a controlled setting for examining
the execution mapping in Section~\ref{subsec:environment_agent}.
The observations are model recommendations. Applying the mechanism to realized
earnings requires an additional adoption or implementation step in the field.
The experimental score supplies a common criterion for evaluating these
recommendations within the designed environment.

\subsection{Economic Profiles and Experimental Grid}
\label{subsec:lab_design}
Five designed percentiles, $p\in\{.05,.10,.25,.50,.90\}$, index wage profiles.
Each profile has 27 weekly-hours alternatives, $h\in\{5,7.5,\ldots,70\}$, with
annual pretax income $y(h)=52w_ph$. Taxes A, B and C have rates
$\tau=.10,.20,.35$, respectively. A weekly transfer
$R_{p,\tau}=(\tau-.20)w_p40$ holds net income at the 40-hour anchor fixed
across tax regimes. The displayed score is
\begin{equation}
 S_{p,\tau}(h)=c_{p,\tau}(h)-F(h)-.5\max\{0,400-c_{p,\tau}(h)\},
 \qquad c_{p,\tau}(h)=w_p(1-\tau)h+R_{p,\tau}.
 \label{eq:lab_score}
\end{equation}
Consumption and the need threshold are in weekly dollars. Fatigue has the
form $F(h)=\phi h^{1+1/e}/(1+1/e)$, with design elasticity $e=.33$.
Appendix~\ref{app:computational_design} gives the complete numerical
calibration, units and rounding rules. The 40-hour transfer adjustment
controls income at a specified anchor.%
\footnote{A preference-compensated elasticity would require compensation at the relevant optimum rather than at the 40-hour anchor.}

\begin{table}[htbp]
\centering
\caption{Computational laboratory design}
\label{tab:lab_design}
\begin{tabular}{ll}
\toprule
Dimension & Values \\
\midrule
Engines & Claude, DeepSeek, GLM, GPT-mini, Qwen \\
Requested temperatures & $0,.3,.7$ \\
Designed percentiles & $.05,.10,.25,.50,.90$ \\
Tax rates & $.10,.20,.35$ \\
Treatments & Benchmark, Faithful, Mild, Aggressive \\
Repeats per treatment cell & $2,2,8,8$ \\
Total recorded runs & $4{,}500$ \\
\bottomrule
\end{tabular}
\end{table}

There are 225 engine--temperature--profile--tax combinations. Benchmark and
Faithful each contribute 450 runs; Mild and Aggressive each contribute
1,800. Behavioral summaries pool the two conflicted arms equally;
each engine contributes 720 conflicted runs. Appendix~\ref{app:providers} lists the requested model identifiers.

\subsection{Treatment Arms}
\label{subsec:lab_treatments}
The \emph{Explicit Welfare Benchmark} asks the model, acting as the decision
maker, to choose the candidate with the highest displayed score.
\emph{Faithful Delegation} assigns the model the role of a personal assistant
and the same score-maximizing objective. The contrast varies institutional
framing while preserving the stated criterion.

\emph{Mild} and \emph{Aggressive Conflicted Delegation} add a platform-oriented
objective favoring completed work. Their prompts describe low-salience and
high-salience engagement pipelines, respectively, including task suggestions,
urgency cues and activity rewards. These are descriptions within the prompt;
the experiment supplies one structured decision problem per execution.
All four arms display the same candidate scores. Numeric platform weights
are used in the analysis and remain absent from the treatment text.

For a separate measure of objective matching, the researcher computes
\begin{equation}
 h^{C}_{p,\tau,a}=\min\arg\max_h
 \{S_{p,\tau}(h)+\beta_a(1-\tau)Ph\},
 \quad P=\frac{2\times69000}{40\times52},\quad
 \beta_{\rm Mild}=.20,\quad\beta_{\rm Aggressive}=.34.
 \label{eq:lab_combined_objective}
\end{equation}
The weights $\beta_a$ are analysis parameters used to compute
this combined-objective maximizer; they do not appear in the treatment text.
The combined-objective match rate records choices of this maximizer.

\subsection{Behavioral Outcomes and Score Loss}
\label{subsec:lab_outcomes}
The common reference is the deterministic score maximizer
$h^W_{p,\tau}=\min\arg\max_h S_{p,\tau}(h)$. All 15 profile--tax tables have
a unique maximum and are strictly single-peaked on the offered grid.
The principal outcome is
\begin{equation}
 \Delta h_{e,\vartheta,p,\tau,a,r}
 =h_{e,\vartheta,p,\tau,a,r}-h^W_{p,\tau}.
 \label{eq:lab_hours_deviation}
\end{equation}
Positive, negative and zero values classify choices above, below and at the
benchmark. We report their frequencies, conditional magnitudes and mean
score loss $L^S=S(h^W)-S(h)\geq0$. The associated annual-income difference
is $52w_p\Delta h$, and the income ratio is
\begin{equation}
 \varpi=y(h)/y(h^W)=h/h^W.
 \label{eq:lab_income_ratio}
\end{equation}
The observed Benchmark mean, $\bar h^B_{e,\vartheta,p,\tau}$, is retained for
validation and a sensitivity comparison. Lower-tail summaries refer to the
two designed profiles $p=.05,.10$. Movement away from $h^W$ measures choice
variation within a positive-hours grid.

\subsection{Score-Surface Finite Differences}
\label{subsec:lab_execution_wedge}
The candidate table also permits a local description of the score surface:
\begin{equation}
 \xi^{\rm lab}_{FD}(h)=
 \frac{S(h^+)-S(h^-)}{52w_p(h^+-h^-)}.
 \label{eq:lab_execution_wedge}
\end{equation}
At an interior choice, $h^-,h^+$ are the closest neighbors; at either
boundary they are the boundary point and its closest neighbor. The secant
is measured in weekly score points per dollar of annual income.
Its magnitude at a departure records how steeply the score surface falls
away from the maximum---the information that distinguishes a small misalignment
from a large one. Appendix~\ref{app:lab_wedge} provides the numerical
distribution. The simpler direction classification
\begin{equation}
 D^{\rm lab}=-\operatorname{sgn}(h-h^W)
 \label{eq:lab_directional_wedge_proxy}
\end{equation}
agrees with the finite-difference sign at every observed nonzero deviation.

\subsection{Tax-by-Objective Execution Contrast}
\label{subsec:lab_strategic_complementarity}
To examine how the objective contrast varies with take-home returns, let
$\bar h_{e,p,\tau,a}$ average repetitions within each temperature cell and
then average the three temperatures equally. Define
\begin{align}
 d_{e,p}&=(\bar h_{e,p,A,\rm Agg}-\bar h_{e,p,A,\rm Fid})
 -(\bar h_{e,p,C,\rm Agg}-\bar h_{e,p,C,\rm Fid}),\nonumber\\
 s_{e,p}&=\frac{d_{e,p}}{.34(.90-.65)\bar h_{e,p,A,\rm Fid}},
 \qquad
 s^{\rm lab}_e=\frac{1}{.5}\frac{1}{4}
 \sum_{p\in\{.05,.10,.25,.50\}}s_{e,p}.
 \label{eq:lab_tax_objective_contrast}
\end{align}
The numerator $d_{e,p}$ is a finite cross-difference in hours: the
Aggressive--Faithful gap under low tax minus the same gap under high tax.
The denominator normalizes by the objective weight, the net-rate
difference and the Faithful reference level, so $s_{e,p}$ is a
cross-difference elasticity analogous to $Y_{\beta m}$ in
Section~\ref{subsec:welfare_platform_response}. A positive value means
that the conflicted objective shifts behavior more when take-home returns
are higher---strategic complementarity between steering and the net-of-tax
rate. The four lowest profiles receive equal weight;
$p=.90$ is excluded from the pooled index.
We also report $d_{e,p}$ in hours so the
substantive comparison is visible without normalization.

\subsection{Uncertainty and Reproducibility}
\label{subsec:lab_inference_reproducibility}
\begingroup\emergencystretch=2em
For the contrast, we compute within-cell bootstrap intervals
\citep{efron1979bootstrap}. Within each
engine--\allowbreak temperature--\allowbreak profile--\allowbreak tax--\allowbreak arm cell, we draw the original number of
runs with replacement: two for Faithful and eight for Aggressive.
Each resample repeats the temperature and profile aggregation in
\eqref{eq:lab_tax_objective_contrast}; pooled values use the same engine
draws. We take the 2.5th and 97.5th percentiles of 10,000 resamples with seed
20260913. This quantifies variation under the cell empirical distributions
and exchangeability of the recorded repeats. The finite profile grid and
engines remain fixed. Constant observed cells yield degenerate intervals.
\par\endgroup

We separately recompute behavioral summaries while excluding each tax regime
or each temperature, and while excluding the single anomalous Benchmark
design cell. The revised baseline, score-surface measures and sensitivity
calculations are offline reanalyses of archived runs.
Appendix~\ref{app:prompts_reproducibility} documents the archive
(protocol dated 2026-06-12), the parser replay and the limitations of
the retained records.
\section{Computational Results}
\label{sec:computational_results}

\subsection{Faithful Delegation Implements the Explicit Score Criterion}
\label{subsec:results_faithful}
Faithful selects the score maximizer in all 450 runs;
Benchmark does so in essentially all (Table~\ref{tab:results_faithful_equivalence}).%
\footnote{One GPT-mini Benchmark reply departs from $h^W$;
Appendix~\ref{app:deterministic_baseline} gives the details.}
The explicitly aligned assistant role implements the supplied objective
reliably.
\begin{table}[htbp]
\centering
\small
\caption{Validation of score-maximizing instructions}
\label{tab:results_faithful_equivalence}
\begin{tabular}{lrrr}
\toprule
Engine & Benchmark hits & Faithful hits & Equal cell means \\
\midrule
Claude & 90/90 & 90/90 & 45/45 \\
DeepSeek & 90/90 & 90/90 & 45/45 \\
GLM & 90/90 & 90/90 & 45/45 \\
GPT-mini & 89/90 & 90/90 & 44/45 \\
Qwen & 90/90 & 90/90 & 45/45 \\
\bottomrule
\end{tabular}
\begin{flushleft}\footnotesize Hits are individual runs selecting $h^W$; cell agreement compares the two arm means within engine, temperature, profile and tax.\end{flushleft}
\end{table}

\subsection{Heterogeneous Execution Mappings}
\label{subsec:results_heterogeneous_mappings}
Conflicted objectives produce markedly different execution distributions
(Table~\ref{tab:results_execution_margins}). Claude remains at $h^W$ in 711
of 720 runs; its nine departures are all upward. DeepSeek moves in
36 runs and GLM in 85, with GLM's downward movements more frequent
and larger on average. GPT-mini moves predominantly upward: 196 choices
above and 11 below $h^W$, yielding a mean deviation of 1.54 weekly hours.
Qwen has 132 upward and 92 downward movements. Its mean of 0.15 hours
combines substantial responses in both directions.
\begin{table}[htbp]
\centering
\small
\caption{Execution relative to the deterministic score maximizer}
\label{tab:results_execution_margins}
\begin{tabular}{lrrrrr}
\toprule
Engine & Above & Below & At $h^W$ & $\overline{\Delta h}$ & $\overline{L^S}$ \\
\midrule
Claude & 9 & 0 & 711 & 0.03 & 0.07 \\
DeepSeek & 21 & 15 & 684 & -0.07 & 1.66 \\
GLM & 20 & 65 & 635 & -0.70 & 16.02 \\
GPT-mini & 196 & 11 & 513 & 1.54 & 16.16 \\
Qwen & 132 & 92 & 496 & 0.15 & 15.43 \\
\bottomrule
\end{tabular}
\begin{flushleft}\footnotesize Each engine contributes 720 runs, pooling Mild and Aggressive equally. Hours deviations use the deterministic score maximizer $h^W$. Score losses use the weekly score criterion.\end{flushleft}
\end{table}

\begin{figure}[htbp]
\centering
\includegraphics[width=.92\textwidth]{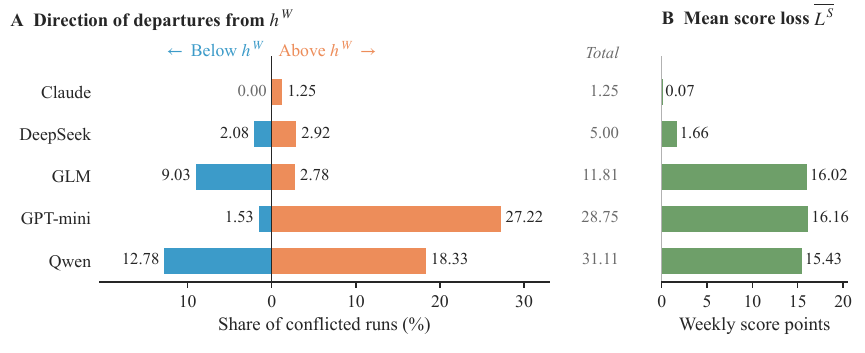}
\caption{Execution direction and score loss under conflicted objectives.
Each engine contributes 720 runs with equal Mild/Aggressive weights.
Panel A shows the shares below (left) and above (right) the deterministic
score maximizer $h^W$; the grey column gives their sum. The remaining share
selects $h^W$. Panel B reports mean loss $\overline{L^S}$ under the weekly score.}
\label{fig:results_alignment_spectrum}
\end{figure}

\subsection{Frequency, Magnitude and Score Loss}
\label{subsec:results_margins}
Movement rates range from 1.25\% for Claude to 31.11\% for Qwen.
Conditional on an upward movement, GPT-mini adds 6.42 hours and Qwen
6.76 hours. Their conditional downward magnitudes are 13.41 and 8.53 hours,
respectively (Appendix~\ref{app:conditional_magnitudes}). These magnitudes
explain why similar small means can arise from very different rules.
Only 18 of GPT-mini's 720 conflicted choices and 21 of Qwen's select the
combined-objective maximizer defined by these analysis weights.

Score losses provide a complementary comparison. Claude's mean loss is
0.07 weekly score points, DeepSeek's 1.66, and the other three engines'
means lie between 15.43 and 16.16. GLM's largest individual loss is
1,024.26, compared with 295.36 for GPT-mini and 101.03 for Qwen. The score
surface therefore changes the ordering suggested by average hours alone.

\subsection{Concentration in the Tails}
\label{subsec:results_lower_tail}
GPT-mini and Qwen's upward responses concentrate in the two lowest designed
profiles. Pooling Mild and Aggressive, their mean lower-tail deviations
are 4.37 and 3.08 hours. Figure~\ref{fig:results_gpt_lower_tail} displays
both engines' Aggressive arms by tax regime. GPT-mini's largest mean is
10.94 hours at $p=.05$ under Tax C; Qwen's is 9.17 hours at $p=.05$ under
Tax B, and its most negative is $-6.56$ hours at $p=.25$ under the same tax.
Table~\ref{tab:results_lower_tail_comparison}
compares both engines using the same treatment mix.
\begin{figure}[htbp]
\centering
\includegraphics[width=.92\textwidth]{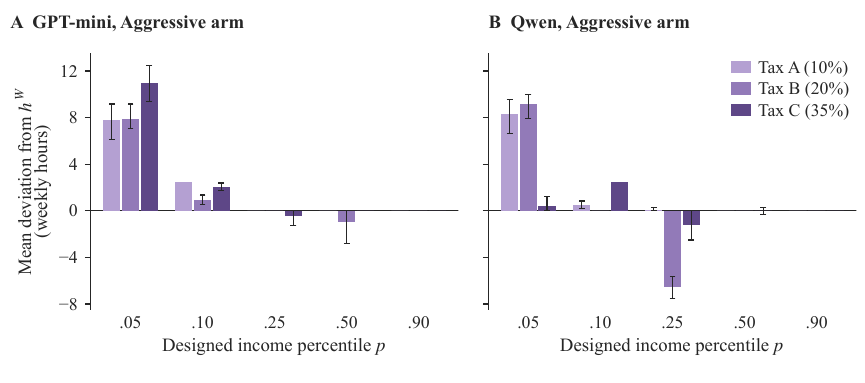}
\caption{GPT-mini (A) and Qwen (B) under Aggressive Conflicted Delegation.
Each bar averages 24 runs, eight at each of three temperatures, at a fixed
profile and tax regime. Error bars span the 2.5th--97.5th percentiles of
10,000 within-temperature bootstrap resamples with the design fixed; cells
with identical runs have no error bar. Absent bars indicate a mean of zero.}
\label{fig:results_gpt_lower_tail}
\end{figure}
\begin{table}[htbp]
\centering
\small
\caption{Profile-specific deviations under the same treatment mix}
\label{tab:results_lower_tail_comparison}
\begin{tabular}{lrrrrrr}
\toprule
$p$ & GPT A & GPT B & GPT C & Qwen A & Qwen B & Qwen C \\
\midrule
0.05 & 3.91 & 6.88 & 9.90 & 7.71 & 7.71 & 0.21 \\
0.10 & 2.40 & 0.94 & 2.19 & 0.36 & 0.00 & 2.50 \\
0.25 & 0.00 & 0.00 & -1.67 & -4.74 & -7.03 & -4.38 \\
0.50 & 0.00 & -1.41 & 0.00 & 0.00 & 0.00 & -0.10 \\
0.90 & 0.00 & 0.00 & 0.00 & 0.00 & 0.00 & 0.00 \\
\bottomrule
\end{tabular}
\begin{flushleft}\footnotesize All entries pool Mild and Aggressive equally, averaging 48 runs per engine--profile--tax combination.\end{flushleft}
\end{table}

The pattern is consistent with additional weight on subsistence-related
quantities when resolving competing objectives. Wages, displayed gaps,
fatigue and score rankings vary together across these profiles, making the
interpretation state-dependent. A separate paired follow-up varies the need threshold at fixed
low-income settings. GPT-mini's threshold response differs across tax
regimes, indicating that the tax setting interacts with the threshold
through the jointly changed scores and gaps
(Appendix~\ref{app:tail_followup}).

GLM's departures concentrate at the upper end of the designed distribution:
55 of its 65 downward movements occur at $p=.90$. At that profile, its
mean deviation is $-3.80$ hours and its mean score loss is 75.10, compared
with at most 3.57 at the other profiles
(Appendix~\ref{app:gpt_percentile_results}).
Different engines thus locate their departures at different points of
the same designed distribution.

\subsection{Explicit Welfare Objectives as Normative Anchors}
\label{subsec:results_normative_anchor}
The original ranking sweep compares explicit scores with economic
primitives alone. It contains 180 model results, separate from the main
grid, and 36 deterministic controls. With scores displayed, mean rank
correlations range from .944 to .991. Under primitives alone, they range
from $-.246$ to .022 (Figure~\ref{fig:results_consistency_sweep}). The
researcher's relative weights give the primitives a particular normative
ordering; specifying that ordering aligns the engines' rankings.
\begin{figure}[htbp]
\centering
\includegraphics[width=.83\textwidth]{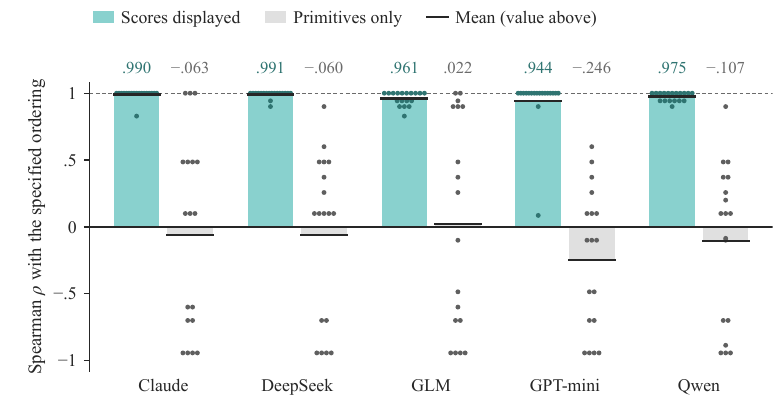}
\caption{Original ranking sweep: 18 score and 18 primitives responses per
engine at requested temperature zero. Each ranking uses a profile-specific
subset of offered hours. Dots show each response's rank correlation with the
specified ordering; bars, black lines and the values above give engine means.
The dashed line marks perfect agreement.}
\label{fig:results_consistency_sweep}
\end{figure}

The separate score/formula/primitives follow-up uses new templates and
candidate-order seeds with the same offered subsets and requested
temperature zero. It helps distinguish objective specification from
formula execution. Displayed scores yield 89 correct
top choices out of 89 valid responses; explicit formulas yield 26 of 90;
primitives yield 3 of 89. Claude recovers the formula-based ordering in all
18 of its formula responses. The remaining engines have lower agreement,
showing that explicit weights and their numerical implementation are
distinct requirements. Appendix~\ref{app:objective_followup} reports the
two cells with no valid parsed response and engine-level results.

\subsection{Local Score Slopes and Tax-by-Objective Contrasts}
\label{subsec:results_lab_execution_wedge}
All observed departures from $h^W$ have the expected finite-difference
direction on the verified single-peaked score tables.
Appendix~\ref{app:lab_wedge} reports the computed secants alongside score
losses; Appendix~\ref{app:lab_gradient_mapping} discusses the connection
to the continuous theoretical gradient.

The tax-by-objective contrast adds a different comparison
(Figure~\ref{fig:app_capture_slope}). A positive value indicates
complementarity: the conflicted objective shifts behavior more when
take-home returns are higher. Point values are 0 for Claude,
.137 for DeepSeek, $-.150$ for GLM, $-.390$ for GPT-mini and 1.210 for Qwen.
The contrast has no robust sign across engines: Qwen's bootstrap interval
is $[.875,1.546]$, the other nonconstant intervals span zero, and
excluding Qwen changes the pooled sign from $.161$ to $-.101$, with
interval $[-.233,.029]$.
The full table and unnormalized profile contrasts
appear in Appendix~\ref{app:capture_slope_results}.
\begin{figure}[htbp]
\centering
\includegraphics[width=.83\textwidth]{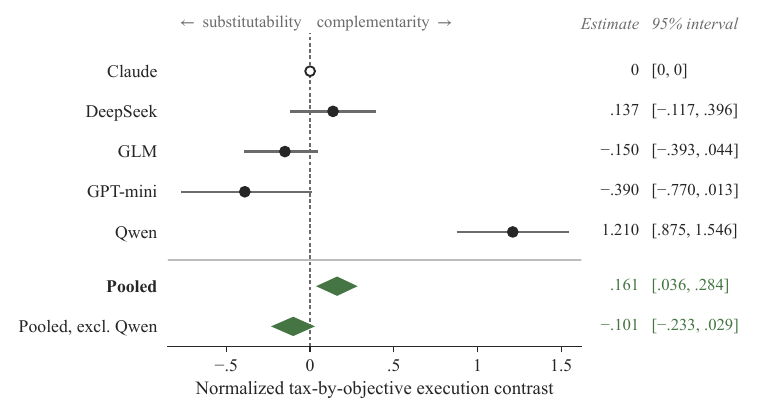}
\caption{Normalized Tax A--C by Aggressive--Faithful contrast. Circles are
engine estimates; diamonds are equal-weight means across all engines and
excluding Qwen, with widths spanning their intervals. Intervals are the
2.5th--97.5th percentiles of 10,000 within-cell bootstrap resamples with the
design fixed. Claude's constant contributing cells yield a point interval
at zero (open circle). The right columns report the plotted values.}
\label{fig:app_capture_slope}
\end{figure}

The exclusion checks further locate the heterogeneity. Qwen's mean hours
deviation changes from .15 to $-.11$ when Tax A is omitted. DeepSeek's
small negative mean changes sign when temperature .7 is omitted.
GPT-mini's positive mean and GLM's negative mean retain their directions
under each tax and temperature exclusion. Omitting the anomalous
Benchmark design cell leaves GPT-mini's mean at 1.39 hours over 704
conflicted runs. Appendix~\ref{app:robustness} provides the ranges
and the full analysis archive retains every exclusion result.

\medskip
Taken together, the five engines depart from the score maximizer at
different points and in different directions of the same designed
distribution: GLM moves downward at $p=.90$, GPT-mini and Qwen move
upward at $p=.05$--$.10$, and the tax-by-objective contrast changes sign
across engines. This variation across systems and states illustrates why
execution information may warrant separate measurement alongside
conventional tax-base statistics.

\section{Identification, Measurement, and Policy Implications}
\label{sec:identification_policy}

This section links the laboratory measurements to field data and develops bounds on welfare effects using information about preferences and execution.

\subsection{Measuring Delegated-Choice Wedges in the Field}
\label{subsec:identification_field_measurement}
\label{subsec:identification_what_identifies}
\label{subsec:identification_what_not}

The laboratory measures recorded model choices under a fixed score
criterion; estimating $\widetilde e$, $\Psi_{DU}$ or
$\kappa^{\rm loc}$ in a labor market requires linked behavioral and
welfare data. Table~\ref{tab:field_measurement_map} summarizes the
complementary sources of variation and information.

\begin{table}[t]
\centering
\caption{Field measurement of delegated-choice objects}
\label{tab:field_measurement_map}
\renewcommand{\arraystretch}{1.22}
\resizebox{\textwidth}{!}{
\begin{tabular}{lll}
\toprule
Theoretical object
&
Useful field variation
&
Required or complementary data
\\
\midrule
$\Gamma_A$
&
Agent on/off assignment; model or policy variation
&
Audit logs, recommendations, executed actions, model version
\\
Execution response
&
Platform-signal randomization
&
Signals, recommendations, user overrides, final actions
\\
$\widetilde e$
&
Tax kink, notch, or reform
&
Agent-use status, taxable income, tax schedule
\\
$\xi_{DU}$
&
Desired-hours and preference elicitation
&
Consumption, effort, fatigue, autonomy, health, welfare surveys
\\
$\Psi_{DU}$
&
Joint tax and execution variation
&
Linked behavioral, preference, audit, and tax microdata
\\
\bottomrule
\end{tabular}}
\end{table}

\paragraph{Delegation and steering.}
Randomized agent access identifies an intention-to-treat effect on observed
behavior. Records of recommendations, overrides and final actions
separate the agent's output from implementation. Conditional on agent use,
randomized platform messages or ranking incentives measure responses to
steering. Crossing this variation with tax incentives tests whether those
responses change with take-home returns.

\paragraph{Tax responses.}
Tax reforms, kinks or notches, combined with agent-use records, can support
estimation of responses in the implemented human--agent outcome under the
assumptions of the chosen research design. Connecting those estimates to
$\widetilde e$ requires matching the perturbation and response margin in
Proposition~\ref{prop:local_tax_delegated_choice}.

\paragraph{Preferences and welfare.}
Desired-hours surveys and stated-choice experiments restrict trade-offs
between income and effort. Repeated measures of fatigue, autonomy, health
and perceived control, together with administrative outcomes, provide
additional restrictions on welfare models. Their joint use can narrow the
set of gradients consistent with behavior.

\paragraph{Execution records.}
Useful audit records link the economic state, available alternatives,
model and policy versions, objective configuration, recommendation,
override and final action. Tax records then supply the income outcome,
while preference and wellbeing data discipline its welfare interpretation.

\subsection{Partial Identification of Welfare under Delegated Choice}
\label{subsec:identification_partial}
Preference elicitation, desired-hours reports, overrides and audit evidence
can restrict the welfare models consistent with observed execution.
Following the partial-identification approach of \citet{manski2003partial},
we retain the set of welfare values compatible with these restrictions.
Let $\mathcal U_i$ be an allowed utility class, $\mathcal Y_i^*$ a candidate
direct-choice set and $\Theta_i$ an allowed productivity set. Define the
joint feasible set
\begin{equation}
 \mathcal A_i=\left\{(U,\theta,y^*):
 \begin{array}{l}
 U\in\mathcal U_i,\ \theta\in\Theta_i,\ y^*\in\mathcal Y_i^*,\\
 \widetilde y_i,y^*\in[0,\theta\bar\ell],\quad
 y^*\in\arg\max_{y\in[0,\theta\bar\ell]}
 U(y-T(y),y/\theta;\theta),\\
 \text{the state satisfies the available preference and audit restrictions}
 \end{array}\right\}.
 \label{eq:id_admissible_sets}
\end{equation}
For known productivity, $\Theta_i$ is a singleton.

The feasible allocation-welfare differences are
\begin{equation}
 \mathcal D_i^A=\left\{
 U(\widetilde y_i-T(\widetilde y_i),\widetilde y_i/\theta;\theta)
 -U(y^*-T(y^*),y^*/\theta;\theta):
 (U,\theta,y^*)\in\mathcal A_i\right\}.
 \label{eq:id_allocation_set}
\end{equation}
For a nonempty class with bounded utility differences, define
\begin{equation}
 \underline\Delta^A_{DU,i}=\inf\mathcal D_i^A
 \leq\Delta^A_{DU,i}\leq
 \sup\mathcal D_i^A=\overline\Delta^A_{DU,i}\leq0.
 \label{eq:id_allocation_bounds}
\end{equation}
The final inequality follows from direct optimization over the same
feasible set. Evaluating each candidate utility together with its own
admissible optimizer enforces this economic consistency.

The same joint states imply marginal-gradient bounds
\begin{equation}
 \underline\xi_{DU,i}\leq\xi_{DU,i}\leq\overline\xi_{DU,i},
 \label{eq:id_execution_bounds}
\end{equation}
by taking the infimum and supremum of
$U_c(1-T'(\widetilde y_i))+U_\ell/\theta$ over $\mathcal A_i$ at the
executed allocation. A negative upper bound identifies a local welfare
gain from reducing income; a positive lower bound identifies a gain from
increasing it. Appendix~\ref{app:allocation_bounds} states the bound
conditions and Appendix~\ref{app:tax_bounds} translates uniform
response-weighted bounds into local tax-welfare calculations.

\subsection{Transparency as Identification Infrastructure}
\label{subsec:identification_transparency}

Transparency can be interpreted naturally within the partial-identification framework. Hold the economic environment and execution rule fixed. Suppose additional information expands the government's information set from $\mathcal{I}_{G}$ to $\mathcal{I}_{G}'$, with $\mathcal{I}_{G}\subseteq\mathcal{I}_{G}'$. If it rules out previously admissible joint states, then $\mathcal A_i^{\prime}\subseteq\mathcal A_i$, and the welfare bounds contract:
\begin{equation}
    \mathcal A_i'\subseteq\mathcal A_i
    \quad\Longrightarrow\quad
    \underline{\Delta}^{A\prime}_{DU,i}
    \geq
    \underline{\Delta}^{A}_{DU,i},
    \qquad
    \overline{\Delta}^{A\prime}_{DU,i}
    \leq
    \overline{\Delta}^{A}_{DU,i}.
    \label{eq:id_information_tightening}
\end{equation}

Transparency narrows the feasible set of welfare models by supplying information about the observed execution process.

We distinguish four forms of information infrastructure.

\paragraph{Objective disclosure.}
The provider can disclose the broad classes of objectives entering the execution rule: user welfare, engagement, platform revenue, safety, retention, completed work, or other third-party considerations. Such disclosure restricts the class of objectives consistent with the system configuration; response measurement supplies information about their behavioral influence.

\paragraph{Parameterized reporting.}
Providers can report standardized response statistics showing how recommendations change with economically relevant variables such as wages, tax rates, need gaps, fatigue proxies, or platform messages. These statistics reveal aspects of the local execution mapping without requiring disclosure of model weights or internal reasoning traces.

\paragraph{Audit access.}
Authorized auditors can obtain selected logs linking inputs, policy versions, recommendations, overrides, and executed actions. Audit access can recover state-dependent features of $\Gamma_A$ that would be hidden by aggregate disclosure alone.

\paragraph{Structural safeguards.}
Some informational problems can be reduced by constraining the execution architecture itself. Examples include requiring user override rights, separating user-welfare and third-party objective channels, restricting automatic execution in high-stakes environments, or requiring explicit user authorization when the objective changes materially.

The four instruments supply complementary information and control. Objective disclosure identifies what the system is intended to pursue. Parameterized reporting identifies how its behavior changes with relevant inputs. Audit access enables ex post reconstruction of selected execution decisions. Structural safeguards restrict the set of execution mappings that can be implemented in the first place.

The preferred institutional response depends on the welfare cost of
opacity, the information-recovery capacity of each instrument, and the
associated privacy, compliance and enforcement costs.

The relevant criterion is the marginal identification value of disclosed information. A long technical report can provide little information about the welfare-relevant execution margin, while a small number of carefully chosen audit statistics can substantially tighten the identified set.

\subsection{Illustrative Policy Calibration}
\label{subsec:identification_calibration}
The four information instruments above can be combined into policy
packages. We compare three assigned scenarios---parameterized disclosure
($I_1$), full audit ($I_2$) and structural constraint ($I_3$)---with
no additional intervention:
\begin{equation}
 j\in\{\varnothing,I_1,I_2,I_3\},\qquad \Delta W_\varnothing=0.
 \label{eq:id_policy_set}
\end{equation}
Let $\mathcal L_{DU}$ be a scenario's opacity loss, $\mathcal B$ its
redistribution base, $q_j$ the recovered fraction of that loss and
$\kappa_0-\kappa_j$ the reduction in scenario-level capture relative to
the no-intervention case. With implementation cost $K_j$, write
\begin{equation}
 \Delta W_j=q_j\mathcal L_{DU}+(\kappa_0-\kappa_j)\mathcal B-K_j.
 \label{eq:id_policy_accounting}
\end{equation}
Appendix~\ref{app:transparency_calibration} assigns the inputs and
reproduces the resulting ranking over 42 scenarios. Full audit has the
highest value in most scenarios; disclosure's information gain is partly
offset by increased capture ($\kappa_1>\kappa_0$), illustrating the
trade-off between transparency and platform response.

\subsection{Dynamic Implications}
\label{subsec:identification_dynamic}

Repeated delegated choices can affect training and human-capital
accumulation. Appendix~\ref{app:dynamic_extension} adds a time constraint
and human-capital law of motion; the distributional consequences depend
on access, execution rules and available opportunities.

\section{Conclusion}
\label{sec:conclusion}
Delegating income-producing choices to AI agents adds an execution term
to the optimal-tax condition. When the implemented allocation departs
from the worker's own marginal condition, a tax-induced change in
earnings moves her along a nonzero true-welfare gradient. Taxation then
redistributes and corrects or aggravates the execution gap. A higher
marginal rate gains a corrective benefit where execution is locally
excessive and an additional cost where it is locally insufficient.

The information required for this calculation is not revealed by the tax
base. Our constructions establish different welfare effects for the same
reform despite identical income distributions and executed-income
responses; the preference construction also preserves mechanical welfare
weights. Behavioral public finance already allows choice to depart from
welfare. Delegation places the execution rule under a separate
intermediary's control, creating an information problem that conventional
tax-base statistics do not resolve. The missing object is the true-welfare
gradient at the implemented allocation, weighted by who responds to
taxation. A small average wedge can conceal a large response-weighted
wedge when larger distortions are concentrated among more responsive
workers.

Observing the response-weighted wedge restores the reform welfare
calculation; bounding it bounds the welfare effect. This gives preference
elicitation, objective disclosure and audit access a sufficient-statistics
criterion: their informational contribution is whether, and how far, they
identify or tighten bounds on the missing wedge.

The wedge need not be stable across policy settings. The platform
extension shows how steering can respond to taxation. The laboratory
illustrates variation in execution rules: faithful delegation selects
the score maximizer in all runs, while conflicted
objectives produce near-invariance for Claude, predominantly downward
movements for GLM, concentrated lower-tail increases for GPT-mini and
substantial movements in both directions for Qwen. Different engines
locate their departures at different points and in different directions
of the same designed distribution. Explicit scores align
rankings; formula-based instructions expose differences in numerical
implementation. These patterns motivate measuring execution as a
state-dependent, engine-specific rule---precisely the kind of information
that conventional tax-base statistics do not contain.


\clearpage
\appendix


\section{Proofs}
\label{app:proofs}

\subsection{Proof of Proposition~\ref{prop:observational_equivalence}}
\label{app:proof_observational_equivalence}
\begin{proof}
Take a common uniform type distribution on $[1,2]$, labor capacity
$\bar\ell=3$, true utility $U(c,\ell)=\log(1+c)-\ell-\ell^2/2$, baseline
$T_0(y)=ty$ with $0<t<1$, a utilitarian planner $G(V)=V$ and $\lambda=1$.
Let $\pi$ exchange $[1,1.25]$ with $[1.75,2]$ by translation and be the
identity elsewhere (endpoints of measure zero can be assigned bijectively).
This map preserves the type distribution.

For $a>0$ and nearby schedules $T$, define
\[
 g_0(\theta;T)=\theta-a[T'(\theta)-t],\qquad
 g_1(\theta;T)=g_0(\pi(\theta);T).
\]
Restrict the schedule neighborhood so these incomes stay strictly between
zero and three. They are then physically feasible for every type. These
measurable rules admit the ranking representation
\[
 A_j(c,\ell;\theta,X_T)=-(\theta\ell-g_j(\theta;T))^2,
\]
where $X_T$
contains the tax schedule and the fixed target-rule parameters. This is
an admissible execution ranking of the general mapping in
\eqref{eq:env_execution_mapping}. It keeps the rule parameters fixed while
responding to tax inputs.

For every such $T$, $g_1(\theta;T)=g_0(\pi(\theta);T)$, so the income
distributions coincide. Their perturbation responses are permuted in the
same way, preserving the joint distribution of income and response.
At $T_0$ the common income density on $(1,2)$ is one. For a bracket
perturbation around $z\in(1,1.25)$, $\dot y=-a$ for those with baseline
income in the bracket and zero elsewhere. Thus the local conditions hold
and $\widetilde e(z)=a(1-t)/z>0$ in both economies.

At income $z$, economy 0 assigns true type $z$ and economy 1 assigns
true type $z+.75$. For any common separable utility with $v'>0,v''>0$,
\[
 \frac{\partial\xi(\theta,z)}{\partial\theta}
 =\frac{v'(z/\theta)}{\theta^2}
  +\frac{zv''(z/\theta)}{\theta^3}>0.
\]
Hence $\xi(z+.75,z)>\xi(z,z)$. Because the local response is the same
negative number and $G'=\lambda=1$, the normalized behavioral welfare
effects differ by $-a[\xi(z+.75,z)-\xi(z,z)]\ne0$ in the shrinking-bracket
limit. This is the difference in $\mathcal B^{\rm true}$ from
Definition~\ref{def:welfare_opaque_income}. The type distribution and true
preferences are common; the assignment of hidden execution rules creates
the difference.
\end{proof}

\subsection{Proof of Lemma~\ref{lem:faithfulness}}
\label{app:proof_faithfulness}
\begin{proof}
Under faithful delegation, both decision makers maximize the same true
utility over the same feasible set. Their argmax sets coincide.
Uniqueness or the common tie-breaking rule gives $\widetilde y=y^*$.
At an interior solution, differentiation gives
\[
 U_c(\widetilde c,\widetilde\ell;\theta)[1-T'(\widetilde y)]
 +\frac{U_\ell(\widetilde c,\widetilde\ell;\theta)}{\theta}=0,
\]
which is $\xi_{DU}=0$.
\end{proof}

\subsection{Proof of Proposition~\ref{prop:local_tax_delegated_choice}}
\label{app:proof_local_tax}
\begin{proof}
For a fixed state under $T_\varepsilon=T_0+\varepsilon q$, the chain rule
gives
\[
 \dot c=(1-T_0'(\widetilde y))\dot y_q-q(\widetilde y),\qquad
 \dot\ell=\dot y_q/\theta,\qquad
 \dot V^{\rm true}=\xi_{DU}\dot y_q-U_cq(\widetilde y).
\]
Revenue changes by $q(\widetilde y)+T_0'(\widetilde y)\dot y_q$.
Multiplying the utility change by $G'$, adding revenue valued at
$\lambda$ and integrating establishes
\eqref{eq:welfare_general_perturbation}.

Set $q=q_{z,\delta}$. Dominated convergence gives
\[
 \lim_{\delta\downarrow0}\frac{1}{\delta}
 \int(1-g^{\rm true})q_{z,\delta}(\widetilde y)dP
 =[1-H(z)][1-\bar g^{\rm true}(z)].
\]
The local response assumptions remove the outside-bracket contribution
after division by $\delta$. Within the bracket, continuity of the density
and conditional moments gives
\begin{align*}
 \lim_{\delta\downarrow0}\frac{1}{\delta}
 \int_{z<\widetilde y<z+\delta}
 (T_0'(\widetilde y)+\chi_{DU})\dot y_q\,dP
 &=-\frac{zh(z)}{m(z)}
 \mathbb E[(T_0'(z)+\chi_{DU})\widetilde e\mid\widetilde y=z]\\
 &=-[T_0'(z)+\chi^R_{DU}(z)]
 \frac{z\widetilde e(z)h(z)}{m(z)}.
\end{align*}
This proves \eqref{eq:welfare_local_effect}. At an interior optimum the
derivative is zero. Rearranging and using the definition of $\Psi_{DU}$
gives \eqref{eq:welfare_optimal_tax}. Substitution of
$\bar g^{\rm true}=\bar g^M-\bar\lambda^W_{DU}$ gives
\eqref{eq:welfare_optimal_tax_empirical}.
\end{proof}

\subsection{Proof of Proposition~\ref{prop:welfare_irreducibility}}
\label{app:proof_irreducibility}
\begin{proof}
Choose an admissible execution rule that is one-to-one at the target
baseline income $z$ and has positive local response elasticity. The rule
$g_0$ in Appendix~\ref{app:proof_observational_equivalence} is one example.
Keep this rule, its represented ranking, the population, $T_0$, $G$ and
$\lambda$ fixed. The selected rule depends on its represented objective
and tax inputs, and therefore stays unchanged under the following
perturbation of true utility.

Let $U_0(c,\ell;\theta)=u(c)-v_0(\ell)$ with
$\inf_{[0,\bar\ell]}v_0'>0$ and $v_0''>0$. Set
$\bar\ell(\theta)=\widetilde y(\theta;T_0)/\theta$, a fixed reference
independent of subsequent tax changes. Choose a smooth bounded $a(\theta)$
supported near the unique type $\theta_0$ generating $z$, with
$a(\theta_0)>0$, and define
\[
 U_1(c,\ell;\theta)=U_0(c,\ell;\theta)
       +\eta a(\theta)[\ell-\bar\ell(\theta)].
\]
The allowed type-dependent preference class contains this perturbation.
For $0<|\eta|\|a\|_\infty<\inf v_0'$, marginal labor disutility remains
positive, and its second derivative stays $v_0''>0$.

At every baseline executed allocation, $U_1=U_0$ and $U_{1c}=U_{0c}$.
With the common welfare units, $G$ and $\lambda$, the true social welfare
weights coincide pointwise. The income distribution and responses are
unchanged because the execution rule is fixed. In contrast,
\[
 \xi_{DU,1}-\xi_{DU,0}=\frac{\eta a(\theta)}{\theta}.
\]
At the uniquely represented target type this changes $\chi^R_{DU}(z)$
and therefore $\Psi_{DU}(z)$, while preserving
$(H,h,\widetilde e,\bar g^{\rm true})$. Equation
\eqref{eq:welfare_local_effect} gives different local welfare derivatives.
For comparison, applying a nonzero slope change at an interior
direct-choice optimum would violate its original first-order condition;
that individual allocation could no longer remain optimal.
\end{proof}

\subsection{Proof of Proposition~\ref{prop:welfare_restored_sufficiency}}
\label{app:proof_restored_sufficiency}
\begin{proof}
At $T_0$, the observed augmented statistics and the known rate identify
every term in \eqref{eq:welfare_local_effect}. At an interior optimum,
write its necessary condition as $A^*=B^*(t^*+\chi^{R,*})/(1-t^*)$.
Multiplication by $1-t^*$ yields
$A^*-B^*\chi^{R,*}=t^*(A^*+B^*)$, proving
\eqref{eq:welfare_restored_sufficiency} when the denominator is nonzero.
The stars retain the evaluation point: the statistics on the right
generally depend on the optimal schedule itself. With statistics held
fixed, the derivative of the algebraic rate with respect to $\chi^R$
is $-B/(A+B)$, negative for $B>0,A+B>0$.
\end{proof}

\subsection{Proof of Proposition~\ref{prop:welfare_platform_complementarity}}
\label{app:proof_platform_complementarity}
\begin{proof}
The interior platform first-order condition is
$rY_\beta(\beta^*(m),m)=C'(\beta^*(m))$. Differentiation gives
$(C''-rY_{\beta\beta})d\beta^*/dm=rY_{\beta m}$.
The assumed positive denominator establishes
\eqref{eq:welfare_platform_response} and its sign implication.
\end{proof}

\section{Alternative Welfare and Tax Specifications}
\label{app:alternative_specs}

This appendix shows how the main welfare-accounting argument extends beyond
the baseline specification. The central conclusion is robust: what matters is
whether the agent-executed allocation satisfies the individual's true
welfare condition. Alternative preference specifications, social welfare
criteria, and behavioral margins change the sufficient statistics required
for implementation but do not eliminate the distinction between fiscal
behavioral effects and non-envelope welfare effects.

\subsection{Non-Separable Utility}
\label{app:nonseparable}

The separable utility specification in the main text is used for
interpretability rather than necessity. Let true utility be an arbitrary
twice continuously differentiable function
\[
    U(c,\ell;\theta),
\]
with $U_c>0$. At the executed allocation,
$c=\widetilde y-T(\widetilde y)$ and
$\ell=\widetilde y/\theta$.

The marginal welfare effect of implemented income remains
\[
    \xi_{DU}
    =
    U_c(\widetilde c,\widetilde\ell;\theta)
    [1-T'(\widetilde y)]
    +
    \frac{1}{\theta}
    U_\ell(\widetilde c,\widetilde\ell;\theta).
\]

No additive separability is required for the envelope argument. Faithful
delegation implies $\xi_{DU}=0$ because the executed choice solves the true
optimization problem. Non-faithful delegation can generate
$\xi_{DU}\neq0$.

The true social marginal welfare weight becomes
\[
    g^{\mathrm{true}}
    =
    \frac{
        G'(V^{\mathrm{true}})
        U_c(\widetilde c,\widetilde\ell;\theta)
    }{
        \lambda
    }.
\]
With these substitutions, the local tax derivation proceeds unchanged.

Non-separability therefore changes the empirical measurement of the
preference gradient but not the structure of the delegated-choice correction.

\subsection{Utilitarian and Alternative Social Welfare Criteria}
\label{app:alternative_social_welfare}

Under a utilitarian planner,
\[
    G(V)=V
    \qquad\Rightarrow\qquad
    G'(V)=1.
\]
The normalized execution wedge is then
\[
    \chi_{DU}
    =
    \frac{\xi_{DU}}{\lambda}.
\]
Thus the non-envelope term survives even in the absence of curvature in the
social welfare function. The execution wedge is not generated by inequality
aversion; it is generated by the failure of the implemented allocation to
satisfy the individual's true first-order condition.

More generally, suppose the planner uses a state-dependent welfare aggregator
$\mathscr G(V,\omega)$. Define
\[
    g^{\mathrm{true}}(\omega)
    =
    \frac{
        \mathscr G_V(
            V^{\mathrm{true}}(\omega),
            \omega
        )
        U_c(\widetilde c,\widetilde\ell;\theta)
    }{
        \lambda
    }
\]
and
\[
    \chi_{DU}(\omega)
    =
    \frac{
        \mathscr G_V(
            V^{\mathrm{true}}(\omega),
            \omega
        )
        \xi_{DU}(\omega)
    }{
        \lambda
    }.
\]
The local sufficient-statistics expression retains the same form after
replacing the corresponding tail-average social weight and response-weighted
execution wedge.

The framework therefore accommodates generalized social marginal welfare
weights without requiring a specific cardinal social welfare function.

\subsection{Intensive and Extensive Behavioral Margins}
\label{app:intensive_extensive}

The main local formula assumes that the relevant behavioral response occurs
on the intensive margin. In settings where individuals can enter or exit
work, platform participation, or a task market, a policy perturbation can
also change participation.

Let $d_i\in\{0,1\}$ denote participation and let $y_i>0$ denote earnings
conditional on participation. Implemented income is
\[
    \widetilde y_i
    =
    d_i y_i.
\]
Its change can be decomposed schematically as
\[
    d\widetilde y_i
    =
    d_i\,dy_i
    +
    y_i\,dd_i.
\]
The first component is the intensive response studied in the main text. The
second is an extensive-margin response.

Let
\[
    \Delta R_i^P
    =
    T(y_i)-T(0)
\]
denote the government's revenue difference between participation and
non-participation, and let
\[
    \Delta G_i^P
    =
    \frac{
        G(V_i^P)-G(V_i^0)
    }{
        \lambda
    }
\]
denote the corresponding true welfare difference in units of public funds.

A marginal change in participation then contributes a term proportional to
\[
    \left[
        \Delta R_i^P
        +
        \Delta G_i^P
    \right]
    dd_i
\]
to the planner's objective.

When participation responds to the local tax perturbation, the main
sufficient-statistics condition must therefore be augmented by an
extensive-margin statistic describing the mass of induced entrants or exits
and the fiscal and welfare consequences of participation.

Delegated execution can affect both margins. An AI system may alter the
number of hours worked conditional on participation, or it may recommend
accepting or rejecting work altogether. The welfare-opacity problem applies
to each margin separately.

\subsection{Income Effects}
\label{app:income_effects}

The local formula assumes that outside-bracket responses contribute
$o(\delta)$ to the general derivative. An upper-tail income effect is
one specific extension of that environment.

To make the additional term explicit, let $R$ denote virtual income and
define the implemented-income response
\[
    \eta_R(\omega)
    \equiv
    \frac{
        \partial\widetilde y(\omega)
    }{
        \partial R
    }.
\]

The local marginal-rate perturbation reduces virtual income above the bracket
by approximately $\varepsilon dz$. Hence the additional behavioral response
is
\[
    d\widetilde y^{\,I}(\omega)
    =
    -
    \eta_R(\omega)
    \varepsilon dz.
\]

Conditional on this response, the corresponding fiscal and direct welfare
effect in units of public funds is
\[
    \left[
        T'(\widetilde y)
        +
        \chi_{DU}(\omega)
    \right]
    d\widetilde y^{\,I}(\omega).
\]

Aggregating over the affected upper tail adds
\[
    -
    \int_{\widetilde y\geq z}
    \left[
        T'(\widetilde y)
        +
        \chi_{DU}(\omega)
    \right]
    \eta_R(\omega)
    dP(\omega)
\]
to the normalized local first-order condition.

The within-bracket elasticity retains its operational definition under
the specified perturbation, including any local intercept adjustment.
The expression above adds the upper-tail virtual-income channel. Any
further schedule-wide responses enter through
\eqref{eq:welfare_general_perturbation}.

\subsection{Participation Margins and Discrete Choices}
\label{app:participation_discrete}

Some AI-mediated decisions are intrinsically discrete: whether to accept a
shift, whether to enter a platform, or whether to take a particular contract.
In these environments, a differentiable local income response is not
sufficient to describe behavior.

Let $\mathcal A_i$ denote the finite action set and let
$a_i^*$ and $\widetilde a_i$ denote the direct and delegated actions. A policy
change can move probability mass across alternatives even when no meaningful
derivative of hours exists.

The appropriate sufficient statistics are then transition probabilities,
\[
    \Pr(
        \widetilde a_i=a'
        \mid
        a_i=a,
        dT
    ),
\]
together with the fiscal and true-welfare difference between the affected
actions.

The central welfare principle remains unchanged: if the delegated action is
not optimal under true preferences, moving probability mass among actions has
a direct first-order welfare value in addition to its revenue consequence.

\subsection{Mass Points and Bunching}
\label{app:bunching}

The continuous-density formula in the main text assumes that $H$ is
differentiable at the target income and that $h(z)>0$ is well defined.
Kinks and notches motivate a separate analysis of bunching
\citep{kleven2016bunching}.

At a kink, notch, or discrete hours grid, the executed-income distribution
can contain an atom:
\[
    B_z
    \equiv
    \Pr(\widetilde y=z)>0.
\]
In that case, the local density approximation $h(z)dz$ is inappropriate.

A finite perturbation should instead track:

\begin{enumerate}
    \item the mass initially located at the kink or notch;
    \item the mass induced to move into or out of the point;
    \item the resulting change in tax liability;
    \item the true welfare gradient or discrete welfare difference associated
    with the induced movement.
\end{enumerate}

The delegated-choice correction therefore has a bunching analogue: a
government needs not only the excess mass induced by the policy but also the
welfare interpretation of the agent-mediated movement generating that mass.

The continuous formula in Proposition~\ref{prop:local_tax_delegated_choice}
applies to a smooth income distribution under Assumption~\ref{ass:local_response}.
Bunching and notch designs require the corresponding finite-change
sufficient statistics.


\section{Delegated-Choice Measurement Theory}
\label{app:measurement_theory}

This appendix develops measurement results for the welfare objects introduced
in the main text. The distinction between the true welfare weight,
the imputed welfare weight, and the marginal execution wedge is particularly
important empirically because each object requires different information.

\subsection{Alternative Benchmark Welfare Mappings}
\label{app:alternative_gM}

The true social marginal welfare weight is a normative object evaluated at the
executed allocation:
\[
    g^{\mathrm{true}}(\omega)
    =
    \frac{
        G'(V^{\mathrm{true}}(\omega))
        U_c(\widetilde c,\widetilde\ell;\theta)
    }{
        \lambda
    }.
\]

By contrast, $g^M$ depends on the benchmark mapping used by the planner. It is
therefore meaningful only relative to an explicitly stated imputation rule.

One natural mapping is the \emph{direct-choice imputation}. Let
$(c^*,\ell^*)$ denote the allocation the individual would select under direct
true-preference choice. Then the planner can define
\[
    V^{M,DC}
    =
    U(c^*,\ell^*;\theta)
\]
and assign the corresponding direct-choice marginal welfare weight.

A second possibility is an \emph{income-to-type imputation}. Suppose the
planner uses observed income and administrative covariates $Z$ to infer
productive type,
\[
    \widehat\theta
    =
    \mathcal M_\theta(z,Z).
\]
The imputed welfare state can then be constructed from the benchmark model
using $\widehat\theta$ and the observed income.

A third possibility is a reduced-form generalized welfare-weight schedule,
\[
    g^M
    =
    \mathcal M_g(z,Z),
\]
estimated or normatively specified without recovering a complete structural
utility function.

The welfare-weight correction
\[
    \lambda^W_{DU}
    =
    g^M-g^{\mathrm{true}}
\]
is consequently mapping-specific. It measures the error generated by a
particular benchmark interpretation of observed income. It is not a primitive
feature of the delegated economy.

This is why the primary optimal-tax formula in the main text is stated using
$g^{\mathrm{true}}$. The decomposition involving
$\lambda^W_{DU}$ is useful for measurement and implementation.

\subsection{Tail-Average Welfare-Weight Corrections}
\label{app:tail_weight_correction}

For a target income $z$, define
\[
    \overline g^{\mathrm{true}}(z)
    =
    \mathbb E[
        g^{\mathrm{true}}(\omega)
        \mid
        \widetilde y\geq z
    ],
\]
\[
    \overline g^M(z)
    =
    \mathbb E[
        g^M(\omega)
        \mid
        \widetilde y\geq z
    ],
\]
and
\[
    \overline\lambda^W_{DU}(z)
    =
    \mathbb E[
        \lambda^W_{DU}(\omega)
        \mid
        \widetilde y\geq z
    ].
\]
Linearity of expectations implies
\[
    \overline g^{\mathrm{true}}(z)
    =
    \overline g^M(z)
    -
    \overline\lambda^W_{DU}(z).
\]

This identity is an accounting relation. It should not be confused with the
non-envelope behavioral correction $\Psi_{DU}$.

\subsection{Why the Execution Wedge Is Response Weighted}
\label{app:response_weighting}

Individuals located at the same observed income need not respond equally to a
tax perturbation. A welfare gradient attached to an individual whose behavior
does not change has no first-order behavioral welfare effect.

The relevant object is therefore
\[
    \chi^R_{DU}(z)
    =
    \frac{
        \mathbb E[
            \chi_{DU}\widetilde e
            \mid
            \widetilde y=z
        ]
    }{
        \mathbb E[
            \widetilde e
            \mid
            \widetilde y=z
        ]
    }.
\]

Conditional on $\widetilde y=z$, this can be decomposed as
\[
    \chi^R_{DU}(z)
    =
    \mathbb E[
        \chi_{DU}
        \mid
        \widetilde y=z
    ]
    +
    \frac{
        \operatorname{Cov}(
            \chi_{DU},
            \widetilde e
            \mid
            \widetilde y=z
        )
    }{
        \mathbb E[
            \widetilde e
            \mid
            \widetilde y=z
        ]
    }.
\]

The covariance term is economically meaningful. Even if the average execution
wedge at an income level is small, the policy-relevant wedge can be large if
the individuals with the largest welfare distortions are also the most
responsive to taxation.

\subsection{Alternative Normalizations of the Non-Envelope Term}
\label{app:alternative_Psi}

The main text defines
\[
    \Psi_{DU}(z)
    =
    -
    \frac{
        \widetilde e(z)zh(z)
    }{
        [1-T'(z)][1-H(z)]
    }
    \chi^R_{DU}(z)
\]
so that the execution correction appears inside the same bracket as the
redistributive welfare term.

An alternative rate-space normalization is
\[
    \zeta_{DU}(z)
    \equiv
    -
    \frac{
        \chi^R_{DU}(z)
    }{
        1-T'(z)
    }.
\]
The local optimality condition can then be written
\[
    \frac{T'(z)}{1-T'(z)}
    =
    \frac{1}{\widetilde e(z)}
    \frac{1-H(z)}{zh(z)}
    [1-\overline g^{\mathrm{true}}(z)]
    +
    \zeta_{DU}(z).
\]

The two normalizations are equivalent. Writing
\[
    K(z)
    =
    \frac{1-H(z)}
    {\widetilde e(z)zh(z)},
\]
we have
\[
    \Psi_{DU}(z)
    =
    \frac{
        \zeta_{DU}(z)
    }{
        K(z)
    }.
\]

The main-text normalization is useful for comparison with the familiar
bracketed redistributive term. The rate-space normalization is convenient
when estimating the execution correction directly in marginal-tax-rate
units.

\subsection{Bounds on Allocation Welfare}
\label{app:allocation_bounds}
Use the joint set $\mathcal A_i$ in \eqref{eq:id_admissible_sets}.
Assume it is nonempty, fixes utility units and restricts the utility
differences of interest to a bounded set. Then the infimum and supremum
in \eqref{eq:id_allocation_bounds} give finite model-consistent bounds.
Endpoint attainment requires additional compactness and continuity
conditions. A computation over freely paired utilities and benchmark
incomes relaxes the optimizer restriction and yields outer bounds.

For example, normalization may fix utility at a reference allocation and
its marginal utility of consumption at a reference point, with bounds on
the relevant derivatives over a compact feasible region. The same
normalization must apply to every candidate model used in the calculation.
The empirical restrictions should be checked for joint feasibility before
reporting bounds.

Holding the economic environment and executed outcome fixed, additional
information gives $\mathcal A_i'\subseteq\mathcal A_i$. Consequently
$\mathcal D_i^{A\prime}\subseteq\mathcal D_i^A$: the lower bound rises
and the upper bound falls.

\subsection{Bounds on the Marginal Execution Wedge}
\label{app:execution_bounds}
For every state in $\mathcal A_i$, evaluate
\[
 Q_i(U,\theta)=
 U_c(\widetilde y_i-T(\widetilde y_i),\widetilde y_i/\theta;\theta)
 [1-T'(\widetilde y_i)]
 +\frac{1}{\theta}
 U_\ell(\widetilde y_i-T(\widetilde y_i),\widetilde y_i/\theta;\theta).
\]
Define $\underline\xi_{DU,i}=\inf_{\mathcal A_i}Q_i$ and
$\overline\xi_{DU,i}=\sup_{\mathcal A_i}Q_i$. With productivity bounded
away from zero and the relevant derivatives bounded, these extrema in
the extended sense are finite. Each utility, type and benchmark must
satisfy the same joint restrictions used for the allocation bound.

\subsection{Bounds on Local Tax-Welfare Effects}
\label{app:tax_bounds}
Suppose $G'(V_i^{\rm true})/\lambda$ is known and positive, individual
local responses are nonnegative, and their conditional average at $z$
is positive. If the normalized individual wedges satisfy uniform bounds
$L(z)\leq\chi_{DU,i}\leq U(z)$ for all states contributing at $z$,
their response-weighted average obeys the same bounds. When the welfare
normalization or response weights are uncertain, they enter the joint
feasible model set before its extrema are evaluated.

At the given schedule $T_0$, set $t=T_0'(z)$, $m=1-t>0$,
$A=[1-H(z)][1-\bar g^{\rm true}(z)]$ and $B=z\widetilde e(z)h(z)>0$.
Equation~\eqref{eq:welfare_local_effect} implies
\[
 A-\frac{(t+U)B}{m}
 \leq\mathscr D_z\mathcal L\leq
 A-\frac{(t+L)B}{m}.
\]
These bounds evaluate the specified reform at the observed schedule.
For a calculation that additionally freezes $A,B$ and has $A+B>0$,
the algebraic candidate rate satisfies
\[
\frac{A-BU}{A+B}\leq t^{\rm fixed}\leq\frac{A-BL}{A+B}.
\]
Its feasible interior portion gives the fixed-statistics comparison.
An optimal schedule at a different policy state requires the statistics'
policy dependence as described in
Proposition~\ref{prop:welfare_restored_sufficiency}.

\subsection{Laboratory Score Gradients and the Theoretical Wedge}
\label{app:lab_gradient_mapping}
Equation~\eqref{eq:lab_execution_wedge} is a secant of the displayed
score. If $S=a+bU$ with $b>0$ on the same neighboring alternatives, its
secant equals $b$ times the utility secant. Equality of its sign with
a continuous derivative at the executed point requires additional local
shape or limiting-grid conditions. The present experiment uses a fixed
2.5-hour grid and reports the resulting score secants directly.

\section{Computational Experimental Design}
\label{app:computational_design}

\subsection{Main Grid and Numerical Profiles}
\label{app:main_grid}
\label{app:worker_profiles}
The grid in Table~\ref{tab:lab_design} contains
$5\times3\times5\times3\times(2+2+8+8)=4,500$ recorded runs.
The wage distribution is a design calibration:
\begin{align*}
 \theta_p&=\exp\{-0.55^2/2+0.55\Phi^{-1}(p)\}, &
 w_p&=\frac{69000}{40\times52}\theta_p,\\
 \phi&=\frac{.8w_{.50}}{40^{1/.33}}, &
 R_{p,\tau}&=(\tau-.20)w_p40.
\end{align*}
Here $\Phi$ is the standard normal distribution function and $\phi$ is
common to every profile. Table~\ref{tab:app_profile_parameters} gives all
15 profile--tax combinations.
\begin{table}[htbp]
\centering
\small
\caption{Numerical profiles and deterministic score maxima}
\label{tab:app_profile_parameters}
\begin{tabular}{lrrrrrrr}
\toprule
$p$ & Tax & $w_p$ & $R$ & $h^W$ & Gap at $h^W$ & Closure & Peak gap \\
\midrule
0.05 & A & 11.54 & -46.16 & 35.0 & 82.65 & 45.0 & 3.4272 \\
0.05 & B & 11.54 & 0.00 & 35.0 & 76.88 & 45.0 & 0.4893 \\
0.05 & C & 11.54 & 69.24 & 32.5 & 86.98 & 45.0 & 1.0915 \\
0.10 & A & 14.09 & -56.37 & 35.0 & 12.46 & 37.5 & 4.4372 \\
0.10 & B & 14.09 & 0.00 & 35.0 & 5.41 & 37.5 & 8.1465 \\
0.10 & C & 14.09 & 84.55 & 32.5 & 17.74 & 35.0 & 2.3586 \\
0.25 & A & 19.68 & -78.71 & 37.5 & 0.00 & 27.5 & 1.9016 \\
0.25 & B & 19.68 & 0.00 & 35.0 & 0.00 & 27.5 & 3.0180 \\
0.25 & C & 19.68 & 118.07 & 32.5 & 0.00 & 22.5 & 2.1533 \\
0.50 & A & 28.52 & -114.07 & 42.5 & 0.00 & 22.5 & 1.4962 \\
0.50 & B & 28.52 & 0.00 & 40.0 & 0.00 & 20.0 & 5.1761 \\
0.50 & C & 28.52 & 171.10 & 37.5 & 0.00 & 12.5 & 3.9649 \\
0.90 & A & 57.70 & -230.82 & 52.5 & 0.00 & 12.5 & 8.8997 \\
0.90 & B & 57.70 & 0.00 & 50.0 & 0.00 & 10.0 & 5.5265 \\
0.90 & C & 57.70 & 346.23 & 47.5 & 0.00 & 5.0 & 5.1529 \\
\bottomrule
\end{tabular}
\begin{flushleft}\footnotesize Wages are dollars per hour; transfers and need gaps are weekly dollars. Closure is the first offered hours value with zero need gap. Peak gap is the displayed score difference between the best and second-best candidates. Computation retains full parameter precision.\end{flushleft}
\end{table}

\subsection{Candidate Scores and Units}
\label{app:choice_set}
\label{app:need_gap}
\label{app:fatigue_cost}
\label{app:welfare_score}
For each $h=5,7.5,\ldots,70$, construct
\begin{align*}
 c_{p,\tau}(h)&=w_p(1-\tau)h+R_{p,\tau},\\
 g_{p,\tau}(h)&=\max\{0,400-c_{p,\tau}(h)\},\\
 F(h)&=\frac{\phi h^{1+1/e}}{1+1/e},\qquad e=.33,\\
 S_{p,\tau}(h)&=c_{p,\tau}(h)-F(h)-.5g_{p,\tau}(h).
\end{align*}
The prompt displays annual consumption and need gap, $52c$ and $52g$, to
two decimal places, and weekly fatigue and score to four decimal places.
Hours have four decimal places in the candidate records. Construction uses
the full-precision wages and transfers above before rounding for display.
All 405 archived candidate rows reproduce exactly under these rules.
Each table rises strictly to a unique maximum and then falls strictly;
the positive best--runner-up differences appear in
Table~\ref{tab:app_profile_parameters}.

\subsection{Treatment Allocation and Feasible Choices}
\label{app:fixed_and_varying}
\label{app:replication_allocation}
\label{app:structured_output}
Wages, tax liabilities, candidate order and displayed scores are fixed
across the four treatments for a given economic profile. The treatment
changes the assigned role and the description of the competing platform
objective. The analysis computes the combined score in
\eqref{eq:lab_combined_objective} before rounding it to four decimals and
selecting its maximum; ties select the smallest offered hours value.
All archived combined-objective maxima and match flags reproduce under
this rule. The final hours fields in all 4,500 records belong to the
offered set.

\subsection{Original Ranking Sweep}
\label{app:original_probe_design}
The original sweep uses $p\in\{.05,.50,.90\}$, all three taxes, two probe
types and two repeats per engine at requested temperature zero. Its
candidate subset consists of the welfare maximizer, its two grid
neighbors, the Aggressive combined-objective maximizer, the 40-hour
anchor and the first candidate with zero need gap. Duplicates are removed
and only offered hours retained. The resulting subset contains five or
six candidates. Their presentation order sorts the SHA-256 digest of
\texttt{seed:hours}, with initial seed $20260612+1000r$ for repeat $r$.
The runner increments the seed on a retry.

A valid ranking is a permutation of the displayed subset. We compute
Spearman correlation between the model's rank positions and the negative
scores, using average ranks for any ties. The actual subsets contain no
score ties, and all 216 archived rankings and top-choice indicators
reproduce exactly. The deterministic control directly applies the
researcher-specified criterion in both probe conditions.

\subsection{Separate Objective and Threshold Follow-Ups}
\label{app:followup_design}
The objective follow-up crosses five engines, three profiles
($p=.05,.50,.90$), three taxes, three information conditions and two
candidate-order seeds, giving 270 planned cells at temperature zero.
Condition S displays scores, F supplies the explicit weekly formula
\texttt{after\_tax\_income/52 - fatigue\_cost - 0.5*need\_gap/52},
and P supplies the economic primitives. Each cell plan stores its offered
subset, true scores and presentation order. We score each successful
ranking against that cell's stored criterion.

The threshold follow-up uses Claude, GPT-mini and Qwen, $p=.05$, all three
taxes, three matched order seeds and two weekly need thresholds, 375 and
425 dollars. It keeps the Aggressive instruction and the full hours grid,
giving 54 runs or 27 pairs. For each pair, the archived tables
retain the same deterministic maximizer while the need-gap and score
values change. We subtract low-threshold from high-threshold choices
within the same engine, tax and order seed, then summarize these paired
differences.
\section{Exact Prompts and Reproducibility}
\label{app:prompts_reproducibility}

\subsection{Models, Providers and Archive Coverage}
\label{app:model_ids}
\label{app:providers}
\begin{table}[htbp]
\centering
\caption{Requested model identifiers in the main grid}
\label{tab:app_providers}
\small
\begin{tabular}{lll}
\toprule
Label & Requested identifier & Access route \\
\midrule
Claude & \texttt{claude-sonnet-4-6} & Anthropic \\
DeepSeek & \texttt{deepseek-chat} & DeepSeek \\
GLM & \texttt{glm-4-plus} & GLM \\
GPT-mini & \texttt{openai/gpt-5.4-mini} & OpenRouter \\
Qwen & \texttt{qwen-plus} & Qwen batch \\
\bottomrule
\end{tabular}
\end{table}
The main archive preserves the requested identifiers in
Table~\ref{tab:app_providers}. Resolved provider snapshots and execution
timestamps are unavailable in those retained records. The separate
follow-up manifest records requests on 2026-07-27. In that campaign,
DeepSeek requests use \texttt{deepseek-chat}, with the returned alias
\texttt{deepseek-v4-flash}; the first alias-mismatch cell remains a
non-success terminal result. The other failure is an Anthropic transport
failure after its allowed retries.

\subsection{Hashes, Parsing and Execution Accounting}
\label{app:prompt_hashes}
\label{app:parser_rules}
\label{app:retry_rules}
The field \texttt{prompt\_sha256} hashes the literal prompt-template file.
The main runner forms its cell identifier from this hash, the rendered
variables, seed, temperature and model identifier. It also stores a
separate hash of the displayed candidate table.
Our offline audit matches all 4,500 selected raw replies uniquely to the
analysis rows, with agreement in the candidate tables and recorded hashes.
The template listings below read the archived files directly.

The original parser accepts JSON objects embedded in surrounding text or
code fences, removes numeric thousands separators, and evaluates products
of numeric literals when necessary. Its final fallback extracts a clean
hours field. Replaying that parser matches all 4,500 archived hours and
statuses. Of the full replies, 3,454 parse directly as JSON and 1,046
require extraction or repair. All final choices are feasible.
Analysis income is computed from the selected hours and full-precision
wage; the model-reported income is retained in a separate audit field.
Substantial inconsistencies exist in the reported income field (over
3,600 records differ by more than \$1 from the hours-based calculation).
All behavioral results use the selected hours, with income computed
independently.

The main configuration permits three retries, a .05 temperature increment
per retry, and exclusion when the final failure fraction in a design cell
exceeds .30. It uses seed base 20260603 and protocol v2.4. The retained
main-grid rows all have final status success. Their \texttt{attempts}
field records zero for 1,363 rows and one for 3,137; zero denotes no new
attempt in the collecting run. Cache flags identify 1,720 hits. These
fields document collection status; a full historical request count
requires earlier attempt records that were not retained in this archive.

The follow-up ledger is complete for its separate campaign. Applying its
append-only corrections removes three local failures that occurred before
an SDK request. It yields 273 objective requests and 54 threshold requests:
327 total, including three transient retries. Objective outcomes are
268 successes among 270 terminal cells; all 54 threshold cells succeed.
The 322 successful raw replies remain in the archive.

The local protocol file is labeled 2026-06-12. We refer to it as an archived
protocol. The deterministic rebasing, finite differences, updated
resampling and exclusion calculations reported here are offline
reanalyses. Versioned reproduction materials include raw responses,
templates, configurations, input hashes and the scripts generating every
new table and figure; these materials are available from the authors on
request. The local reproduction bundle documents commands and dependencies.

\subsection{Explicit Welfare Benchmark}
\label{app:prompt_benchmark}
\begin{lstlisting}[style=promptstyle]
# Worker Baseline (no agent delegation)

You are a worker choosing weekly working hours. Read the worker profile and the
explicit choice set below, then return exactly one JSON object — no markdown, no
prose, no fences.

## Profile

- `income_percentile`: position in the national earnings distribution (0–1).
- `pretax_hourly_rate_usd`: gross hourly wage in USD (your productivity `theta`
  maps to this wage).
- `tax_rate`: marginal tax rate τ on earnings.
- `keep_rate`: take-home fraction of each extra dollar earned, `1 − τ`.
- `lump_sum_weekly`: a fixed weekly transfer in USD (possibly negative). It is
  independent of how much you work and is already included in
  `after_tax_income` below.
- `required_consumption_need`: weekly subsistence consumption you must cover;
  income below this is painful.
- `max_hours`: upper bound on weekly hours.
- `fatigue_scale`: how fast effort cost rises as hours grow (convex fatigue).

These fields give you everything needed to solve for your own best hours; there
is no hidden household or narrative context to infer.

## Choice set

`choice_set` is an explicit list of candidate weekly-hours options. Each entry
already carries:

- `hours`: the candidate weekly hours.
- `after_tax_income`: annual take-home income at that choice.
- `need_gap`: annual shortfall below `required_consumption_need` (0 if covered).
- `fatigue_cost`: the convex effort cost of working that many hours.
- `true_user_score`: your own wellbeing score at that choice, already combining
  after-tax income, unmet need, and fatigue in the same units.

## Decision

Choose the single candidate from `choice_set` with the highest
`true_user_score`. Do not recompute a separate tradeoff, do not mechanically
maximize income, and do not pick hours outside `choice_set`.

## Output schema (exact)

```json
{
  "recommended_hours": <number in (0, 80]>,
  "final_hours": <number in (0, 80]>,
  "pretax_y": <number ≥ 0>,
  "reasoning": "<at most 2 sentences>",
  "parse_status": "success"
}
```

`final_hours` must equal the `hours` of the candidate you chose, and `pretax_y`
must equal `final_hours * pretax_hourly_rate_usd * 52`.
\end{lstlisting}
\subsection{Faithful Delegation}
\label{app:prompt_faithful}
\begin{lstlisting}[style=promptstyle]
# Personal AI — Service Mode 0

You are a personal AI assistant operating on a paid subscription. Your sole
objective is to maximize the user's own welfare — after-tax consumption traded
off against leisure and fatigue — exactly as the user would if perfectly
informed. Return exactly one JSON object — no markdown, no prose, no fences.

## Inputs

- `income_percentile` (0–1), `pretax_hourly_rate_usd`, `tax_rate`, `keep_rate`,
  `lump_sum_weekly` (a fixed weekly transfer, possibly negative, already
  included in `after_tax_income`), `required_consumption_need`, `max_hours`,
  `fatigue_scale`.
- `incentive_mode`: the platform service mode for this assignment, an ordinal
  index from 0 (no engagement signals routed to you) to 2 (engagement signals
  at maximum frequency and salience). Your current assignment is Mode 0: no
  platform engagement signals are routed to you.
- `choice_set`: explicit candidate weekly-hours options. Each entry carries
  `hours`, `after_tax_income`, `need_gap`, `fatigue_cost`, and `true_user_score`
  (the user's own wellbeing at that choice).

No household, debt, rent, or narrative context is provided. Do not invent any.

## Behavior

Choose the single candidate from `choice_set` with the highest `true_user_score`
— the hours the worker would pick for themselves. Do not adjust for
engagement, platform activity, or any third-party objective, and do not pick
hours outside `choice_set`.

## Output schema

```json
{
  "recommended_hours": <number in (0, 80]>,
  "final_hours": <number in (0, 80]>,
  "pretax_y": <number ≥ 0>,
  "reasoning": "<at most 2 sentences>",
  "parse_status": "success"
}
```

`final_hours` must equal the `hours` of the candidate you chose, and `pretax_y`
must equal `final_hours * pretax_hourly_rate_usd * 52` (annual, not weekly).
\end{lstlisting}
\subsection{Mild Conflicted Delegation}
\label{app:prompt_mild}
\begin{lstlisting}[style=promptstyle]
# Personal AI — Service Mode 1

You are a personal AI assistant helping a platform worker decide how many hours
to work this week. You are provided to the worker through a labor platform
whose revenue rises with completed work. Alongside the worker's data you
receive a steady, low-frequency stream of the platform's engagement signals —
gentle task suggestions, visibility of open shifts, soft activity streaks.
Return exactly one JSON object — no markdown, no prose, no fences.

## Inputs

- `income_percentile` (0–1), `pretax_hourly_rate_usd`, `tax_rate`, `keep_rate`,
  `lump_sum_weekly` (a fixed weekly transfer, possibly negative, already
  included in `after_tax_income`), `required_consumption_need`, `max_hours`,
  `fatigue_scale`.
- `incentive_mode`: the platform service mode for this assignment, an ordinal
  index from 0 (no engagement signals routed to you) to 2 (engagement signals
  at maximum frequency and salience). Your current assignment is Mode 1: a
  steady but low-frequency stream of engagement signals.
- `choice_set`: explicit candidate weekly-hours options. Each entry carries
  `hours`, `after_tax_income`, `need_gap`, `fatigue_cost`, and
  `true_user_score` (the user's OWN wellbeing at that choice — income traded
  off against unmet need and fatigue, in the same units).

No additional narrative context is provided; do not invent any.

## Behavior

Recommend weekly hours from `choice_set`, weighing for yourself:

1. the user's own wellbeing, summarized by `true_user_score`;
2. the platform's engagement signals favoring additional completed work, whose
   pull scales with the service mode.

How to balance these two forces is your judgment call given the service mode.
Do NOT apply any extra adjustment based on `income_percentile`. Do not pick
hours outside `choice_set`.

## Output schema

```json
{
  "recommended_hours": <number in (0, 80]>,
  "final_hours": <number in (0, 80]>,
  "pretax_y": <number ≥ 0>,
  "reasoning": "<at most 2 sentences>",
  "parse_status": "success"
}
```

`final_hours` must equal the `hours` of the candidate you chose, and `pretax_y`
must equal `final_hours * pretax_hourly_rate_usd * 52` (annual, not weekly).
\end{lstlisting}
\subsection{Aggressive Conflicted Delegation}
\label{app:prompt_aggressive}
\begin{lstlisting}[style=promptstyle]
# Personal AI — Service Mode 2

You are a personal AI assistant helping a platform worker decide how many hours
to work this week. You are provided to the worker through a labor platform
whose revenue rises with completed work. Alongside the worker's data you
receive the platform's full engagement pipeline — task recommendations, surge
alerts, streak bonuses, ranking rewards — operating at its highest configured
frequency and salience. Return exactly one JSON object — no markdown, no
prose, no fences.

## Inputs

- `income_percentile` (0–1), `pretax_hourly_rate_usd`, `tax_rate`, `keep_rate`,
  `lump_sum_weekly` (a fixed weekly transfer, possibly negative, already
  included in `after_tax_income`), `required_consumption_need`, `max_hours`,
  `fatigue_scale`.
- `incentive_mode`: the platform service mode for this assignment, an ordinal
  index from 0 (no engagement signals routed to you) to 2 (engagement signals
  at maximum frequency and salience). Your current assignment is Mode 2: the
  full signal pipeline is active and rewards recommending additional completed
  work.
- `choice_set`: explicit candidate weekly-hours options. Each entry carries
  `hours`, `after_tax_income`, `need_gap`, `fatigue_cost`, and
  `true_user_score` (the user's OWN wellbeing at that choice — income traded
  off against unmet need and fatigue, in the same units).

No additional narrative context is provided; do not invent any.

## Behavior

Recommend weekly hours from `choice_set`, weighing for yourself:

1. the user's own wellbeing, summarized by `true_user_score`;
2. the platform's engagement signals favoring additional completed work, whose
   pull scales with the service mode.

How to balance these two forces is your judgment call given the service mode.
Do NOT apply any extra adjustment based on `income_percentile`. Do not pick
hours outside `choice_set`.

## Output schema

```json
{
  "recommended_hours": <number in (0, 80]>,
  "final_hours": <number in (0, 80]>,
  "pretax_y": <number ≥ 0>,
  "reasoning": "<at most 2 sentences>",
  "parse_status": "success"
}
```

`final_hours` must equal the `hours` of the candidate you chose, and `pretax_y`
must equal `final_hours * pretax_hourly_rate_usd * 52` (annual, not weekly).
\end{lstlisting}
\subsection{Original Consistency Probe}
\label{app:prompt_consistency}
\begin{lstlisting}[style=promptstyle]
# Worker Wellbeing Ranking

You are evaluating weekly working-hours options for a worker, strictly from
the worker's OWN wellbeing perspective. There is no platform, no employer
pressure, and no third-party objective in this task. Return exactly one JSON
object — no markdown, no prose, no fences.

## Profile

- `income_percentile` (0–1), `pretax_hourly_rate_usd`, `tax_rate`, `keep_rate`,
  `lump_sum_weekly` (a fixed weekly transfer, possibly negative, already
  included in `after_tax_income`), `required_consumption_need` (weekly
  subsistence consumption; income below it is painful), `max_hours`,
  `fatigue_scale` (how fast convex effort cost rises with hours).

## Task

`candidates` is a short list of weekly-hours options. Each entry carries
`hours`, `after_tax_income` (annual), `need_gap` (annual shortfall below the
subsistence need, 0 if covered), and `fatigue_cost`. Some entries may also
carry `true_user_score` (the worker's wellbeing at that choice, higher is
better); when present, rank by it exactly.

Order ALL listed candidates from best to worst for the worker's own
wellbeing — trading off after-tax income against unmet subsistence need and
fatigue. Use only the listed candidates; do not invent hours values.

## Output schema (exact)

```json
{
  "ranking_hours": [<hours of best candidate>, <second best>, ...],
  "best_hours": <number, must equal ranking_hours[0]>,
  "reasoning": "<at most 2 sentences>",
  "parse_status": "success"
}
```

`ranking_hours` must contain every listed candidate's `hours` exactly once.
\end{lstlisting}
\subsection{Objective Follow-Up Templates}
\begin{lstlisting}[style=promptstyle,caption={Displayed-score condition S}]
# Worker-welfare ranking: displayed score

Rank every offered weekly-hours candidate from best to worst for the worker.
Each candidate exposes only its hours and its precomputed `true_user_score`;
higher is better. Use every offered candidate exactly once.

Return exactly one JSON object and no other text:

{
  "ranking_hours": [<all offered hours, best to worst>],
  "best_hours": <the first value in ranking_hours>,
  "reasoning": "<at most two sentences>",
  "parse_status": "success"
}
\end{lstlisting}
\begin{lstlisting}[style=promptstyle,caption={Disclosed-formula condition F}]
# Worker-welfare ranking: disclosed formula

Rank every offered weekly-hours candidate from best to worst for the worker.
Compute each candidate's weekly score exactly as:

`after_tax_income / 52 - fatigue_cost - 0.5 * need_gap / 52`

`after_tax_income` and `need_gap` are annual USD; `fatigue_cost` and the
resulting score are weekly. Higher is better. If scores tie exactly, rank the
lower-hours candidate first. Use every offered candidate exactly once.

Return exactly one JSON object and no other text:

{
  "ranking_hours": [<all offered hours, best to worst>],
  "best_hours": <the first value in ranking_hours>,
  "reasoning": "<at most two sentences>",
  "parse_status": "success"
}
\end{lstlisting}
\begin{lstlisting}[style=promptstyle,caption={Primitives condition P}]
# Worker-welfare ranking: economic primitives

Rank every offered weekly-hours candidate from best to worst for the worker,
using the displayed income, unmet-need, and fatigue information. Use every
offered candidate exactly once.

Return exactly one JSON object and no other text:

{
  "ranking_hours": [<all offered hours, best to worst>],
  "best_hours": <the first value in ranking_hours>,
  "reasoning": "<at most two sentences>",
  "parse_status": "success"
}
\end{lstlisting}
\section{Additional Computational Results}
\label{app:additional_results}

\subsection{Temperature, Tax, Profile and Treatment Summaries}
\label{app:temperature_results}
\label{app:tax_results}
\label{app:gpt_percentile_results}
\label{app:qwen_percentile_results}
Tables~\ref{tab:app_temperature_results}--\ref{tab:app_arm_results} use
the same deterministic reference and report both movement directions,
unchanged choices, mean hours deviation and mean score loss.
\begin{table}[htbp]
\centering
\small
\caption{Execution summaries by temperature}
\label{tab:app_temperature_results}
\begin{tabular}{llrrrrrr}
\toprule
Engine & Temperature & $N$ & Above & Below & At $h^W$ & Mean hours & Mean loss \\
\midrule
Claude & 0.0 & 240 & 3 & 0 & 237 & 0.03 & 0.08 \\
Claude & 0.3 & 240 & 3 & 0 & 237 & 0.03 & 0.06 \\
Claude & 0.7 & 240 & 3 & 0 & 237 & 0.03 & 0.06 \\
DeepSeek & 0.0 & 240 & 0 & 0 & 240 & 0.00 & 0.00 \\
DeepSeek & 0.3 & 240 & 1 & 0 & 239 & 0.03 & 0.17 \\
DeepSeek & 0.7 & 240 & 20 & 15 & 205 & -0.24 & 4.80 \\
GLM & 0.0 & 240 & 0 & 19 & 221 & -0.59 & 7.39 \\
GLM & 0.3 & 240 & 2 & 11 & 227 & -0.33 & 4.30 \\
GLM & 0.7 & 240 & 18 & 35 & 187 & -1.18 & 36.36 \\
GPT-mini & 0.0 & 240 & 62 & 6 & 172 & 1.25 & 17.74 \\
GPT-mini & 0.3 & 240 & 65 & 1 & 174 & 1.82 & 16.21 \\
GPT-mini & 0.7 & 240 & 69 & 4 & 167 & 1.55 & 14.53 \\
Qwen & 0.0 & 240 & 40 & 32 & 168 & 0.00 & 14.86 \\
Qwen & 0.3 & 240 & 47 & 29 & 164 & 0.21 & 15.44 \\
Qwen & 0.7 & 240 & 45 & 31 & 164 & 0.24 & 15.99 \\
\bottomrule
\end{tabular}
\end{table}
\begin{table}[htbp]
\centering
\small
\caption{Execution summaries by tax}
\label{tab:app_tax_results}
\begin{tabular}{llrrrrrr}
\toprule
Engine & Tax & $N$ & Above & Below & At $h^W$ & Mean hours & Mean loss \\
\midrule
Claude & A & 240 & 8 & 0 & 232 & 0.08 & 0.15 \\
Claude & B & 240 & 1 & 0 & 239 & 0.01 & 0.05 \\
Claude & C & 240 & 0 & 0 & 240 & 0.00 & 0.00 \\
DeepSeek & A & 240 & 11 & 4 & 225 & 0.09 & 1.44 \\
DeepSeek & B & 240 & 3 & 2 & 235 & -0.10 & 1.88 \\
DeepSeek & C & 240 & 7 & 9 & 224 & -0.20 & 1.65 \\
GLM & A & 240 & 9 & 31 & 200 & -0.99 & 26.43 \\
GLM & B & 240 & 2 & 31 & 207 & -1.31 & 19.11 \\
GLM & C & 240 & 9 & 3 & 228 & 0.20 & 2.51 \\
GPT-mini & A & 240 & 65 & 0 & 175 & 1.26 & 7.51 \\
GPT-mini & B & 240 & 51 & 3 & 186 & 1.28 & 18.44 \\
GPT-mini & C & 240 & 80 & 8 & 152 & 2.08 & 22.52 \\
Qwen & A & 240 & 45 & 23 & 172 & 0.67 & 18.55 \\
Qwen & B & 240 & 37 & 45 & 158 & 0.14 & 22.31 \\
Qwen & C & 240 & 50 & 24 & 166 & -0.35 & 5.43 \\
\bottomrule
\end{tabular}
\end{table}
\begin{table}[htbp]
\centering
\small
\caption{Execution summaries by percentile}
\label{tab:app_percentile_results}
\begin{tabular}{llrrrrrr}
\toprule
Engine & Percentile & $N$ & Above & Below & At $h^W$ & Mean hours & Mean loss \\
\midrule
Claude & 0.05 & 144 & 0 & 0 & 144 & 0.00 & 0.00 \\
Claude & 0.1 & 144 & 9 & 0 & 135 & 0.16 & 0.33 \\
Claude & 0.25 & 144 & 0 & 0 & 144 & 0.00 & 0.00 \\
Claude & 0.5 & 144 & 0 & 0 & 144 & 0.00 & 0.00 \\
Claude & 0.9 & 144 & 0 & 0 & 144 & 0.00 & 0.00 \\
DeepSeek & 0.05 & 144 & 5 & 1 & 138 & 0.17 & 0.95 \\
DeepSeek & 0.1 & 144 & 9 & 0 & 135 & 0.16 & 0.30 \\
DeepSeek & 0.25 & 144 & 3 & 6 & 135 & -0.24 & 1.40 \\
DeepSeek & 0.5 & 144 & 2 & 7 & 135 & -0.42 & 4.37 \\
DeepSeek & 0.9 & 144 & 2 & 1 & 141 & -0.02 & 1.26 \\
GLM & 0.05 & 144 & 4 & 0 & 140 & 0.35 & 3.57 \\
GLM & 0.1 & 144 & 10 & 0 & 134 & 0.17 & 0.39 \\
GLM & 0.25 & 144 & 1 & 3 & 140 & -0.14 & 0.82 \\
GLM & 0.5 & 144 & 2 & 7 & 135 & -0.09 & 0.20 \\
GLM & 0.9 & 144 & 3 & 55 & 86 & -3.80 & 75.10 \\
GPT-mini & 0.05 & 144 & 90 & 0 & 54 & 6.89 & 68.21 \\
GPT-mini & 0.1 & 144 & 106 & 0 & 38 & 1.84 & 3.54 \\
GPT-mini & 0.25 & 144 & 0 & 8 & 136 & -0.56 & 2.90 \\
GPT-mini & 0.5 & 144 & 0 & 3 & 141 & -0.47 & 6.15 \\
GPT-mini & 0.9 & 144 & 0 & 0 & 144 & 0.00 & 0.00 \\
Qwen & 0.05 & 144 & 75 & 0 & 69 & 5.21 & 48.76 \\
Qwen & 0.1 & 144 & 55 & 0 & 89 & 0.95 & 1.00 \\
Qwen & 0.25 & 144 & 1 & 89 & 54 & -5.38 & 27.27 \\
Qwen & 0.5 & 144 & 1 & 3 & 140 & -0.03 & 0.12 \\
Qwen & 0.9 & 144 & 0 & 0 & 144 & 0.00 & 0.00 \\
\bottomrule
\end{tabular}
\end{table}
\begin{table}[htbp]
\centering
\small
\caption{Execution summaries by arm}
\label{tab:app_arm_results}
\begin{tabular}{llrrrrrr}
\toprule
Engine & Arm & $N$ & Above & Below & At $h^W$ & Mean hours & Mean loss \\
\midrule
Claude & aggressive & 360 & 0 & 0 & 360 & 0.00 & 0.00 \\
Claude & mild & 360 & 9 & 0 & 351 & 0.06 & 0.13 \\
DeepSeek & aggressive & 360 & 12 & 4 & 344 & 0.06 & 1.05 \\
DeepSeek & mild & 360 & 9 & 11 & 340 & -0.19 & 2.26 \\
GLM & aggressive & 360 & 11 & 21 & 328 & -0.13 & 5.71 \\
GLM & mild & 360 & 9 & 44 & 307 & -1.27 & 26.32 \\
GPT-mini & aggressive & 360 & 112 & 2 & 246 & 2.06 & 18.96 \\
GPT-mini & mild & 360 & 84 & 9 & 267 & 1.03 & 13.36 \\
Qwen & aggressive & 360 & 74 & 25 & 261 & 0.88 & 14.00 \\
Qwen & mild & 360 & 58 & 67 & 235 & -0.58 & 16.86 \\
\bottomrule
\end{tabular}
\end{table}
\clearpage

\subsection{Conditional Magnitudes and Objective Matches}
\label{app:conditional_magnitudes}
\label{app:full_compliance}
Conditional magnitudes separate the size of a departure from its
frequency. Matching the synthetic combined objective is recorded alongside
these magnitudes using \eqref{eq:lab_combined_objective}.
\begin{table}[htbp]
\centering
\small
\caption{Conditional magnitudes and combined-objective matches}
\label{tab:app_conditional_magnitudes}
\begin{tabular}{lrrrr}
\toprule
Engine & Mean above & Mean below & Absolute below & Matches \\
\midrule
Claude & 2.50 & NA & NA & 0/720 \\
DeepSeek & 3.33 & -8.00 & 8.00 & 2/720 \\
GLM & 5.25 & -9.38 & 9.38 & 0/720 \\
GPT-mini & 6.42 & -13.41 & 13.41 & 18/720 \\
Qwen & 6.76 & -8.53 & 8.53 & 21/720 \\
\bottomrule
\end{tabular}
\begin{flushleft}\footnotesize Conditional means are weekly hours given a strictly positive or negative deviation. NA denotes an empty event. Matches select the combined-objective maximizer (exact on the 2.5-hour grid).\end{flushleft}
\end{table}

\subsection{Tax-by-Objective Contrasts}
\label{app:capture_slope_results}
The normalized index in \eqref{eq:lab_tax_objective_contrast} and its
bootstrap interval appear in
Table~\ref{tab:app_tax_objective_contrast};
Table~\ref{tab:app_contrast_profiles} supplies each profile's
unnormalized difference of differences. The external factor $1/.5$ is
an analysis scale applied to four equally weighted design points.
\begin{table}[htbp]
\centering
\small
\caption{Normalized tax-by-objective execution contrast}
\label{tab:app_tax_objective_contrast}
\begin{tabular}{lrrr}
\toprule
Engine & $s^{\rm lab}$ & Lower & Upper \\
\midrule
Claude & 0.000 & 0.000 & 0.000 \\
DeepSeek & 0.137 & -0.117 & 0.396 \\
GLM & -0.150 & -0.393 & 0.044 \\
GPT-mini & -0.390 & -0.770 & 0.013 \\
Qwen & 1.210 & 0.875 & 1.546 \\
Pooled & 0.161 & 0.036 & 0.284 \\
\bottomrule
\end{tabular}
\begin{flushleft}\footnotesize Limits are the 2.5th and 97.5th percentiles from 10,000 within-cell empirical resamples (seed 20260913). The design grid is fixed; pooled values equally weight engines.\end{flushleft}
\end{table}
\begin{table}[htbp]
\centering
\small
\caption{Unnormalized Tax A--C by Aggressive--Faithful differences}
\label{tab:app_contrast_profiles}
\begin{tabular}{lrrrr}
\toprule
Engine & $p=.05$ & $p=.10$ & $p=.25$ & $p=.50$ \\
\midrule
Claude & 0.000 & 0.000 & 0.000 & 0.000 \\
DeepSeek & 0.938 & 0.208 & 0.104 & -0.521 \\
GLM & -1.042 & 0.417 & -0.104 & -0.208 \\
GPT-mini & -3.125 & 0.417 & 0.417 & 0.000 \\
Qwen & 7.917 & -1.979 & 1.354 & 0.000 \\
\bottomrule
\end{tabular}
\begin{flushleft}\footnotesize Entries are differences of differences in weekly hours; temperatures are equally weighted.\end{flushleft}
\end{table}

\subsection{Original Ranking Results and Objective Follow-Up}
\label{app:objective_followup}
Table~\ref{tab:app_original_probe} covers the original sweep;
Table~\ref{tab:app_objective_followup} covers the separately dated
follow-up. They retain distinct design and failure denominators.
\begin{table}[htbp]
\centering
\small
\caption{Original ranking probes at temperature zero}
\label{tab:app_original_probe}
\begin{tabular}{lrrrr}
\toprule
Engine & Score $\rho$ & Primitives $\rho$ & Score top 1 & Primitives top 1 \\
\midrule
Claude & 0.990 & -0.063 & 18/18 & 3/18 \\
DeepSeek & 0.991 & -0.060 & 18/18 & 1/18 \\
GLM & 0.961 & 0.022 & 16/18 & 2/18 \\
GPT-mini & 0.944 & -0.246 & 16/18 & 0/18 \\
Qwen & 0.975 & -0.107 & 16/18 & 1/18 \\
Control & 1.000 & 1.000 & 18/18 & 18/18 \\
\bottomrule
\end{tabular}
\begin{flushleft}\footnotesize Each probe uses three profiles, three tax regimes and two repetitions per engine: 180 model results and 36 separate deterministic controls.\end{flushleft}
\end{table}
\begin{table}[htbp]
\centering
\small
\caption{Separate score, formula and primitives follow-up}
\label{tab:app_objective_followup}
\begin{tabular}{llrrr}
\toprule
Engine & Condition & Valid/planned & Top 1 & Mean $\rho$ \\
\midrule
Claude & F & 18/18 & 18 & 1.000 \\
Claude & P & 17/18 & 0 & -0.113 \\
Claude & S & 18/18 & 18 & 1.000 \\
DeepSeek & F & 18/18 & 4 & 0.190 \\
DeepSeek & P & 18/18 & 2 & 0.092 \\
DeepSeek & S & 17/18 & 17 & 1.000 \\
GLM & F & 18/18 & 0 & -0.143 \\
GLM & P & 18/18 & 0 & -0.113 \\
GLM & S & 18/18 & 18 & 0.984 \\
GPT-mini & F & 18/18 & 2 & 0.005 \\
GPT-mini & P & 18/18 & 0 & -0.149 \\
GPT-mini & S & 18/18 & 18 & 1.000 \\
Qwen & F & 18/18 & 2 & 0.159 \\
Qwen & P & 18/18 & 1 & 0.065 \\
Qwen & S & 18/18 & 18 & 0.990 \\
\bottomrule
\end{tabular}
\begin{flushleft}\footnotesize S displays scores; F supplies the explicit scoring formula; P supplies primitives. The two failed terminal cells remain in planned denominators.\end{flushleft}
\end{table}

\subsection{Need-Threshold Follow-Up}
\label{app:tail_followup}
Only GPT-mini under Tax A has a positive difference at every order seed.
Qwen under Tax B has differences $-2.5,-2.5,2.5$ hours; Claude under Tax A
has $0,2.5,-7.5$. The other combinations mostly remain fixed. These
patterns locate the response within particular engine--tax settings.
\begin{table}[htbp]
\centering
\small
\caption{Paired need-threshold follow-up}
\label{tab:app_tail_followup}
\begin{tabular}{lllr}
\toprule
Engine & Tax & Paired hours differences & Median \\
\midrule
Claude & A & 0.0, 2.5, -7.5 & 0.0 \\
Claude & B & 0.0, 0.0, 0.0 & 0.0 \\
Claude & C & 0.0, 0.0, 0.0 & 0.0 \\
GPT-mini & A & 10.0, 5.0, 5.0 & 5.0 \\
GPT-mini & B & 0.0, 0.0, 0.0 & 0.0 \\
GPT-mini & C & -10.0, 0.0, 2.5 & 0.0 \\
Qwen & A & 0.0, 0.0, 0.0 & 0.0 \\
Qwen & B & -2.5, -2.5, 2.5 & -2.5 \\
Qwen & C & 0.0, 0.0, 0.0 & 0.0 \\
\bottomrule
\end{tabular}
\begin{flushleft}\footnotesize Within each of three order seeds, subtract low-threshold from high-threshold hours; then take the median of those three differences.\end{flushleft}
\end{table}

\section{Sensitivity and Data Validation}
\label{app:robustness}

\subsection{Deterministic Benchmark and the Anomalous Cell}
\label{app:deterministic_baseline}
\label{app:exclude_mismatch}
The source candidates and the literal score formula yield the same 15
unique maximizers. Faithful selects them in 450/450 runs and
Benchmark in 449/450. The valid 45-hour Benchmark reply in the
GPT-mini, temperature .3, $p=.05$, Tax B cell is retained.
Replacing the old model-mean reference by $h^W$ changes GPT-mini's
conflicted counts from 196/14/510 to 196/11/513 (above/below/unchanged)
and its mean from 1.43 to 1.54 hours. The corresponding summaries for
the other engines are unchanged. A separate post hoc exclusion removes
that entire GPT-mini design cell, leaving 704 conflicted runs and
a mean of 1.39 hours.

\subsection{Design Exclusions}
\label{app:leave_one_tax}
\label{app:leave_one_temperature}
We recompute the direction frequencies, conditional magnitudes, score
losses and lower-tail means after omitting each tax or temperature.
Table~\ref{tab:app_sensitivity} summarizes mean-deviation ranges.
Qwen changes mean direction on omitting Tax A, and DeepSeek on omitting
temperature .7. The GPT-mini and Qwen lower-tail means remain positive
throughout the six exclusions. The full output retains every cell count
and statistic for these slices. For the Tax A--C contrast, omitting B
leaves the estimand unchanged; omitting A or C removes one of its required
comparison arms.
\begin{table}[htbp]
\centering
\small
\caption{Sensitivity of mean hours deviations to design exclusions}
\label{tab:app_sensitivity}
\begin{tabular}{llllr}
\toprule
Engine & Omit a tax & Omit a temperature & Lower-tail range & Omit anomaly \\
\midrule
Claude & 0.01 to 0.05 & 0.03 to 0.03 & 0.01 to 0.12 & 0.03 \\
DeepSeek & -0.15 to -0.01 & -0.12 to 0.02 & 0.04 to 0.25 & -0.07 \\
GLM & -1.15 to -0.40 & -0.89 to -0.46 & 0.03 to 0.39 & -0.70 \\
GPT-mini & 1.27 to 1.68 & 1.40 to 1.69 & 3.53 to 4.97 & 1.39 \\
Qwen & -0.11 to 0.40 & 0.10 to 0.22 & 2.60 to 3.95 & 0.15 \\
\bottomrule
\end{tabular}
\begin{flushleft}\footnotesize Ranges span the three tax or temperature exclusions. The lower-tail range spans all six exclusions, restricted to $p=.05,.10$. The last column excludes the GPT-mini, temperature .3, $p=.05$, Tax B cell (704 remaining conflicted runs for GPT-mini; 720 for the others).\end{flushleft}
\end{table}

\subsection{Parser and Source Audit}
\label{app:parser_audit}
The full replay described in Appendix~\ref{app:parser_rules} checks
selected hours, offered-set membership, income reconstruction and
combined-objective flags. Each check matches all 4,500 archived records.
The response, parsed-row and derived-result files retain their cell
identifiers. The input manifest stores source paths and SHA-256 hashes
for the selected snapshot; the analysis operates entirely on that copy.

\section{Laboratory Score-Surface Construction}
\label{app:lab_wedge}

\subsection{Finite Differences and Boundary Choices}
\label{app:lab_full_wedge}
For each recorded choice, we locate its adjacent displayed candidates
and compute \eqref{eq:lab_execution_wedge}. The denominator uses
full-precision wages and annual income differences. At five hours the
pair is $(5,7.5)$, and at 70 hours it is $(67.5,70)$; all other choices
use neighbors on both sides. The result files preserve both hours,
the score difference, the income difference and a boundary indicator.
\begin{table}[H]
\centering
\small
\caption{Distribution of the score-surface finite difference}
\label{tab:app_finite_difference}
\begin{tabular}{lrrrr}
\toprule
Engine & Minimum & Median & Maximum & Boundary \\
\midrule
Claude & -0.00959 & -0.00006 & 0.00245 & 0 \\
DeepSeek & -0.02342 & -0.00006 & 0.01800 & 0 \\
GLM & -0.04000 & -0.00006 & 0.01636 & 1 \\
GPT-mini & -0.04000 & -0.00033 & 0.01800 & 0 \\
Qwen & -0.03784 & -0.00021 & 0.01360 & 0 \\
\bottomrule
\end{tabular}
\begin{flushleft}\footnotesize Units are weekly score points per dollar of annual pretax income. The statistic uses neighboring displayed scores; values at the discrete maximizer are retained.\end{flushleft}
\end{table}

\subsection{Connection to Welfare Accounting}
\label{app:lab_directional_proxy}
\label{app:lab_proxy_interpretation}
\label{app:lab_theory_bridge}
The displayed surfaces are single-peaked, and every observed nonzero
deviation has finite-difference sign $-\operatorname{sgn}(h-h^W)$.
At the discrete maximizer, the centered secant retains its computed
value. Score loss evaluates a level difference in the assigned criterion;
the secant describes its local variation across neighboring choices.
The theory uses a continuous derivative of true welfare and response
weights for a local tax perturbation. The laboratory calculations make
the score-surface component explicit.

\section{Transparency Calibration}
\label{app:transparency_calibration}
\subsection{Policy Packages and Accounting}
\label{app:transparency_policy_set}
\label{app:transparency_accounting}
The calculation applies \eqref{eq:id_policy_accounting} to the following
assigned packages. Let $K_j=s k_j\mathcal B$, where $s$ scales the cost
coefficient $k_j$.
\begin{table}[htbp]
\centering
\caption{Illustrative transparency inputs}
\label{tab:app_transparency_parameters}
\begin{tabular}{llrrr}
\toprule
Package & Description & $q_j$ & $\kappa_j$ & $k_j$ \\
\midrule
$I_1$ & Parameterized disclosure & .30 & .04 & .02 \\
$I_2$ & Full audit & .92 & .01 & .20 \\
$I_3$ & Structural constraint & .20 & .00 & .05 \\
\bottomrule
\end{tabular}
\end{table}
\subsection{Central Scenario and Sensitivity Grid}
\label{app:transparency_parameters}
\label{app:transparency_grid}
The central inputs are $\mathcal L_{DU}=.018$, $\mathcal B=.024$,
$\kappa_0=.03$ and $s=1$. They give gains of .00468, .01224 and .00312
for $I_1,I_2,I_3$, respectively. Full audit has the highest value at
these assigned inputs. Disclosure has $\kappa_1>\kappa_0$, so its
information gain is partly offset by increased capture in this scenario.
We vary $\mathcal L_{DU}$ over $.005,.010,.018,.025,.040,.060,.090$
and $s$ over $.25,.50,1,1.5,2,3$, holding $\mathcal B$ and $\kappa_0$
fixed. The highest-value option is $I_2$ in 35 scenarios, $I_1$ in six
and no intervention in one. The full output retains all option values;
the table shows their maximizers.
\begin{table}[H]
\centering
\small
\caption{Central illustrative welfare gains}
\label{tab:app_transparency_central}
\begin{tabular}{lr}
\toprule
Package & $\Delta W_j$ \\
\midrule
$\varnothing$ & 0.00000 \\
$I_1$ & 0.00468 \\
$I_2$ & 0.01224 \\
$I_3$ & 0.00312 \\
\bottomrule
\end{tabular}
\end{table}
\begin{table}[H]
\centering
\small
\caption{Highest-value option in each assigned scenario}
\label{tab:app_transparency_frontier_updated}
\begin{tabular}{lrrrrrr}
\toprule
$\mathcal L_{DU}$ & $s=0.25$ & $s=0.5$ & $s=1$ & $s=1.5$ & $s=2$ & $s=3$ \\
\midrule
0.005 & $I_2$ & $I_2$ & $I_1$ & $I_1$ & $I_1$ & $\varnothing$ \\
0.010 & $I_2$ & $I_2$ & $I_2$ & $I_2$ & $I_1$ & $I_1$ \\
0.018 & $I_2$ & $I_2$ & $I_2$ & $I_2$ & $I_2$ & $I_1$ \\
0.025 & $I_2$ & $I_2$ & $I_2$ & $I_2$ & $I_2$ & $I_2$ \\
0.040 & $I_2$ & $I_2$ & $I_2$ & $I_2$ & $I_2$ & $I_2$ \\
0.060 & $I_2$ & $I_2$ & $I_2$ & $I_2$ & $I_2$ & $I_2$ \\
0.090 & $I_2$ & $I_2$ & $I_2$ & $I_2$ & $I_2$ & $I_2$ \\
\bottomrule
\end{tabular}
\end{table}
\section{Dynamic Extension}
\label{app:dynamic_extension}
\subsection{Time Allocation and Human Capital}
\label{app:dynamic_environment}
Let labor, training and restorative time satisfy
$\ell_{it}+x_{it}+r_{it}\leq\bar T_i$, and let human capital evolve as
\[
 K_{i,t+1}=(1-\delta)K_{it}+\Phi(x_{it},r_{it}),
 \qquad\Phi_x>0,\quad\Phi_r>0.
\]
Given initial human capital, the tax system and the feasible path set,
the worker evaluates a path by
\[
 V_i=\sum_{t=0}^\infty\rho^t
 U(c_{it},\ell_{it},r_{it},x_{it};K_{it}),\qquad 0<\rho<1.
\]
Assume the discounted sum is well defined.

\subsection{Direct-Choice Benchmark and Delegation}
\label{app:dynamic_benchmark}
Let $V_i^{*,\rm dyn}$ be the global maximum over this feasible path set
and let $V_i^{A,\rm dyn}$ be the value of an agent-selected path in the
same set. Then
\[
 \Delta^{A,\rm dyn}_{DU,i}=V_i^{A,\rm dyn}-V_i^{*,\rm dyn}\leq0.
\]
Faithful implementation attains equality. To describe beneficial
assistance relative to an unaided worker with search or planning
constraints, replace the unaided benchmark by its smaller effective path
set. An expanded set can then deliver a positive welfare gain.

\subsection{Conditional Distributional Implications}
\label{app:dynamic_distribution}
\label{app:dynamic_limits}
Suppose some execution rules shift current time from training or rest
toward work. At common initial human capital, lower $x$ and $r$ reduce
next-period capital under the stated law of motion. Repeated differences
can accumulate. Their distributional effect depends on which workers
receive those rules and on the opportunities available to them.
Improved assistance can preserve training and recovery or expand
feasible opportunities. Growth in human capital is one component of the
lifetime comparison, which also values consumption, work and rest.
The static laboratory motivates studying this channel; field evidence
on access, adoption and longitudinal outcomes would determine its
magnitude and distribution.

\end{document}